\documentclass[runningheads]{llncs}
   \usepackage[utf8]{inputenc}
   \usepackage{dutchcal}

\usepackage{tikz-cd}
\usepackage{comment}
\usepackage{graphicx}
\usepackage{longtable}
\usepackage{amsfonts}
\usepackage{amssymb}
\usepackage{bm}
\usepackage[utf8]{inputenc}
\usepackage[english]{babel}
\usepackage[fleqn]{amsmath}
\usepackage{xstring}
\usepackage{setspace}
\usepackage{array}
\usepackage{subfig}
\usepackage{enumitem}
\usepackage{multirow}
\usepackage{ctable} 

\newcommand{\FunFail}[1]{\langle #1 \rangle} 
\newcommand{\FunMode}[2]{[ #1 ]_{#2}}
\newcommand\Defs {{:=}\,}
\newcommand\Wide[1] {~~~#1~~~}
\newcommand\True {\textsf{t}}
\newcommand\False {\textsf{f}}
\newcommand\modes {\sf md}

\newcommand{\RevFunMode}[2]{[ #1 ]^{-}_{#2}}

\newcommand{\Var}[1] {#1} 

\newcommand{\ValRl}[1]{
    \IfEqCase{#1}{
        {x}{\alpha}
        {y}{\beta}
        {z}{\gamma}
        {w}{\delta}
        {v}{\zeta}
        {#1}{#1}
    }
}
\newcommand{\ValRlV}[1]{
    \IfEqCase{#1}{
        {x}{\bm{\alpha}}
        {y}{\bm{\beta}}
        {z}{\bm{\gamma}}
        {w}{\bm{\delta}}
        {v}{\bm{\zeta}}
        {#1}{\bm{#1}}
    }
}
\newcommand{\ValRp}[1]{
    \IfEqCase{#1}{
        {x}{\tilde{\alpha}}
        {y}{\tilde{\beta}}
        {z}{\tilde{\gamma}}
        {w}{\tilde{\delta}}
        {v}{\tilde{\zeta}}
        {#1}{\tilde{#1}}
    }
}
\newcommand{\ValRpV}[1]{
    \IfEqCase{#1}{
        {x}{\bm{\tilde{\alpha}}}
        {y}{\bm{\tilde{\beta}}}
        {z}{\bm{\tilde{\gamma}}}
        {w}{\bm{\tilde{\delta}}}
        {v}{\bm{\tilde{\zeta}}}
        {#1}{\bm{\tilde{#1}}}
    }
}
\newcommand{\FVal}[1]{ 
 {\Hat{#1}}
 }   
\newcommand{\FValV}[1]{
    \IfEqCase{#1}{
        {x}{\bm{\dot{\alpha}}}
        {y}{\bm{\dot{\beta}}}
        {z}{\bm{\dot{\gamma}}}
        {w}{\bm{\dot{\delta}}}
        {v}{\bm{\dot{\zeta}}}
        {#1}{\bm{\dot{#1}}}
    }
}

\newcommand{\Imply}{\Rightarrow}

\newcommand{\Scenario}[3]{\{~ #1~\} ~~#2~~ \{~#3~\}}

\begin{document}
\raggedbottom
\title{Reasoning with failures}
\author{Hamid Jahanian\inst{1} \and Annabelle McIver\inst{2}}

\institute{\email{hamid.jahanian@hdr.mq.edu.au}\and
 \email{annabelle.mciver@mq.edu.au}}
\maketitle 

\begin{abstract}
Safety Instrumented Systems (SIS) protect major hazard facilities, e.g. power plants, against catastrophic accidents. An SIS consists of hardware components and a controller software -- the ``program''. Current safety analyses of SIS' include the construction of a fault tree, summarising potential faults of the components and how they can arise within an SIS.  The exercise of identifying faults typically relies on the experience of the safety engineer. Unfortunately the program part is often too complicated to be analysed in such a ``by hand" manner and so the impact it has on the resulting safety analysis is not accurately captured. In this paper we demonstrate how a formal model for faults and failure modes can be used to analyse the impact of an SIS program. We outline the underlying concepts of \emph{Failure Mode Reasoning} and its application in safety analysis, and we illustrate the ideas on a practical example.
\end{abstract}

\section{Introduction}
Plant accidents can have catastrophic consequences. An explosion at a chemical plant in eastern China killed over 70 people and injured more than 600 in 2019. Safety Instrumented Systems (SIS) are protection mechanisms against major plant accidents \cite{Ref_190}.  Failure of SIS components can result in the SIS being unavailable to respond to hazardous situations.
It is therefore crucial to analyse and address such failures. A typical SIS comprises physical components to interact with plant, and a software program\footnote{In this paper the term \textit{program} refers to the software code run by SIS CPU; also known in safety standards as SIS Application Program \cite{Ref_190}.} that analyses the information and initiates safety actions. Such software can be highly complex, and even when it is not itself faulty still propagate input faults from the sensors to the safety actuators. This paper concerns a current omission in the standard safety engineering process: that of an accurate fault analysis of complex SIS program. 

Well established methods, such as Fault Tree Analysis (FTA), already exist in the industry for analysing and quantifying SIS failure modes \cite{Ref_176}. FTA is a deductive method that uses fault trees for quantitative and qualitative analysis of failure scenarios. A fault tree is a graphical representation of the conditions that contribute to the occurrence of a predefined failure event. A fault tree will be created by a safety analyst and based on their knowledge and understanding of the failure behaviours in a system. Not only are such by-hand analyses inherently subject to human-error, they also require expertise, time and effort. 

A new method, Failure Mode Reasoning (FMR), was recently introduced to circumvent the need for by-hand analysis of parts of SIS \cite{Ref_200}. Using a special calculus built on failure modes, FMR analyses the SIS program to identify the hardware faults at SIS inputs that can result in a given failure at its output. The main outcome of FMR is a short list of failure modes, which can also be used to calculate the probability of failure. In this paper we show how to use ideas from  formal methods to justify FMR. We use an abstraction to model failures directly, and we show that such an abstraction can be used to track failures within the SIS program so that potential output failures can be linked to the potential input failures that cause them. We prove the soundness of the technique and illustrate it on a practical example.

The rest of this paper is organised as follows: Section \ref{bkgsec} provides a brief explanation of the context and how FMR can enhance safety analysis. Section \ref{s1245} formalises the underlying ideas of analysis of failures for SIS programs. Based on these concepts, Sections \ref{fmrsec} and \ref{rfbssec} formulate the concepts for composing the individual elements in FMR and the reasoning process on the interactions between these elements. Section \ref{rflsec} includes descriptions of how FMR is applied in practice and in particular in large scale projects. Finally Sections \ref{discsec} and \ref{concsec} wrap up the paper with a review of FMR's position with respect to other research works and potential research in future.

\section{SIS and FMR}\label{bkgsec}
An SIS consists of sensors, a logic solver, and final elements. The sensors collect data about the environment (such as temperature and pressure) and the logic solver processes the sensor readings and controls the final elements to intervene and prevent a hazard. Such interventions can include shutting down the (industrial) plant and they are referred to as Safety Instrumented Functions (SIFs). Fig.~\ref{SimpFTAB} illustrates a simple SIS consisting of two sensors, one logic solver and one final element. This SIS performs only one SIF, which is to protect the downstream process against high pressure in the upstream gas pipe. The sensors measure the gas pressure and the logic solver initiates a command to close the valve if the gas pressure exceeds a threshold limit. 

\begin{figure}[!ht]
\centering
\subfloat[Fault Tree\label{SimpFTAA}]{\includegraphics[width=0.28\columnwidth]{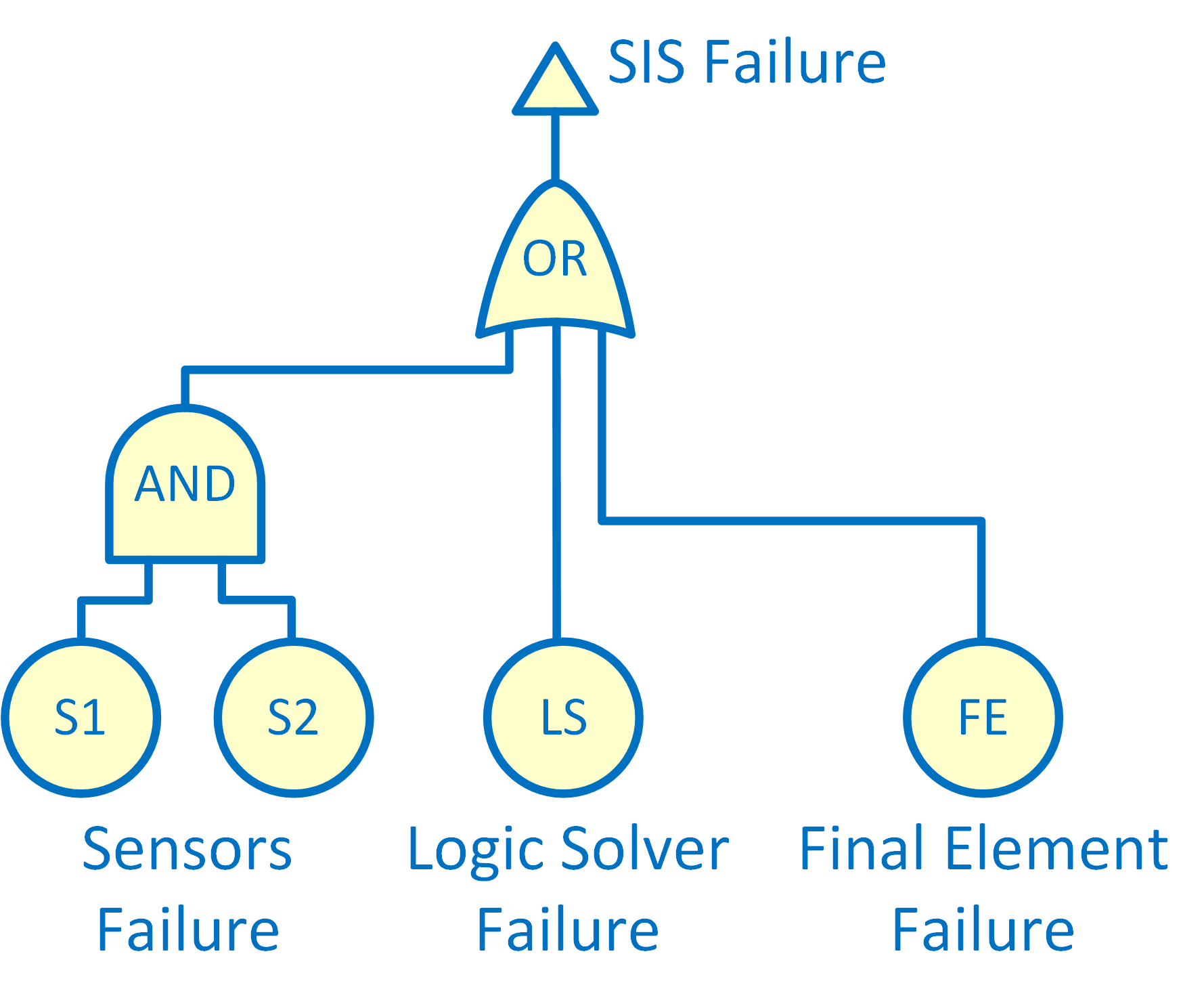}}
\qquad
\subfloat[SIS\label{SimpFTAB}]{\includegraphics[width=0.425\columnwidth]{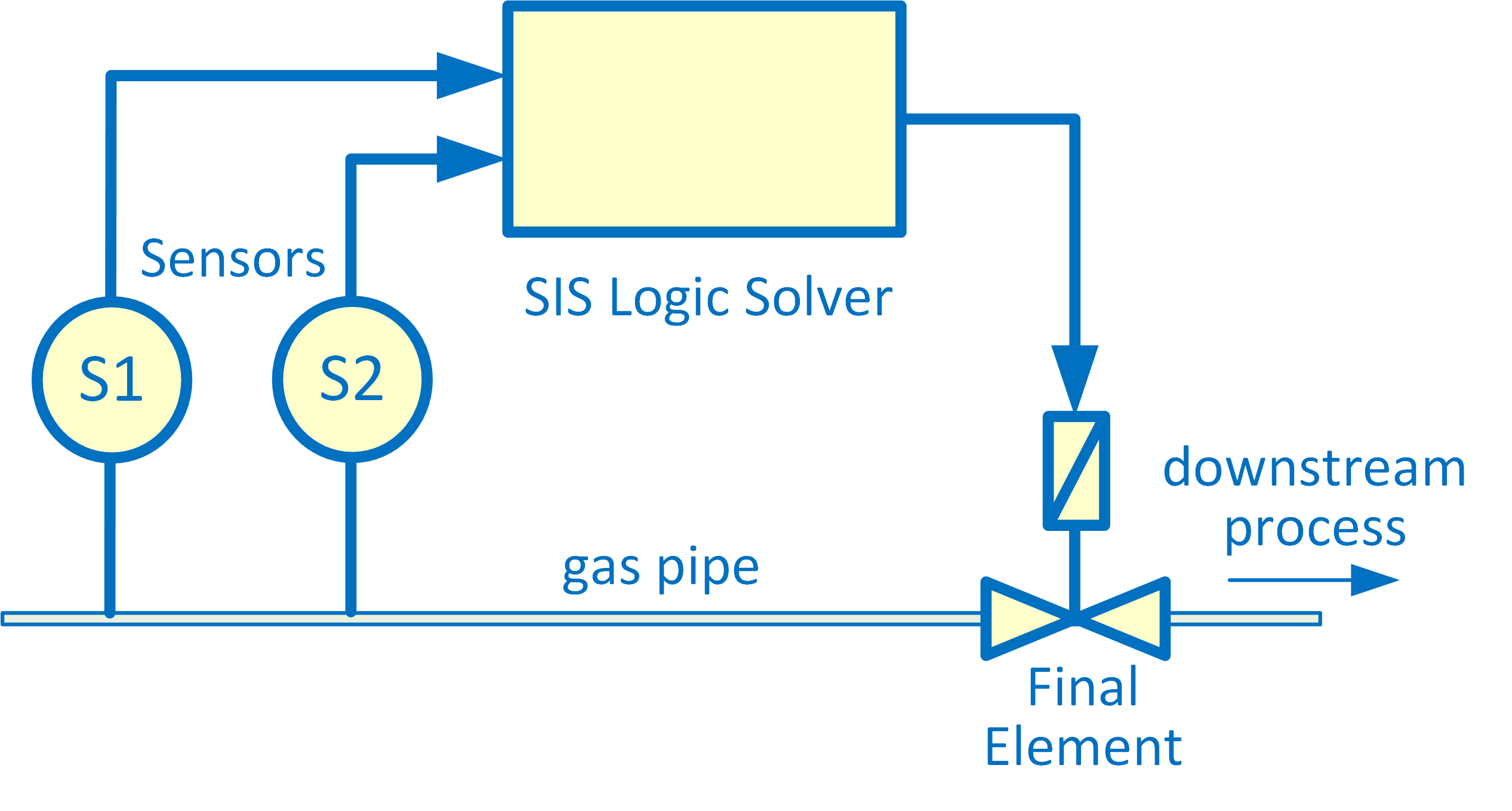}}
\caption{An example SIS and its corresponding fault tree model}
\label{SimpFTA}
\end{figure}

SIS faults are typically modelled by using fault trees. For an accurate analysis a fault tree must reflect all potential faults caused by all components in the SIS. Clearly incorrect sensor readings are a significant factor in safety analysis as they can lead to hazardous scenarios. One of the problems in safety analysis is to understand how such deviations can be propagated by the SIS program and lead to faults at SIS outputs. If done by hand, such understandings depends critically on the analyst's knowledge of the details of the SIS program. 

Consider, for example the fault tree in Fig.~\ref{SimpFTAA}, which is meant to summarise the failures of SIS in  Fig.~\ref{SimpFTAB}: the SIS fails if both sensors fail or if the logic solver fails or if the final element fails. The fault tree is built on the assumption that the two sensors provide redundancy, which means that provided one of the two sensors is in a healthy state, that is sufficient to detect potential hazards. However, the validity of this assumption, and thus the validity of the fault tree, directly depends on the details of SIS program and how it computes the output from the input. For example, if the two inputs from sensors are averaged first and then compared to the high pressure limit as shown in Fig.~\ref{SimpProgA}, the proposed fault tree (Fig.~\ref{SimpFTAA}) is incorrect; because failure of one sensor will affect the average of the two. But if each sensor reading is separately compared to the threshold limit first (as in Fig.~\ref{SimpProgB}), the sensors can be considered redundant and the fault tree would summarise the failures accurately. While the two programs deliver the same functionality, they do not show the same failure behaviour; and the proposed fault tree can correspond to only one of them. 

\begin{figure}[!ht]
{\centering
\subfloat[Program $T_{Avg}$\label{SimpProgA}]{\includegraphics[width=0.43\columnwidth]{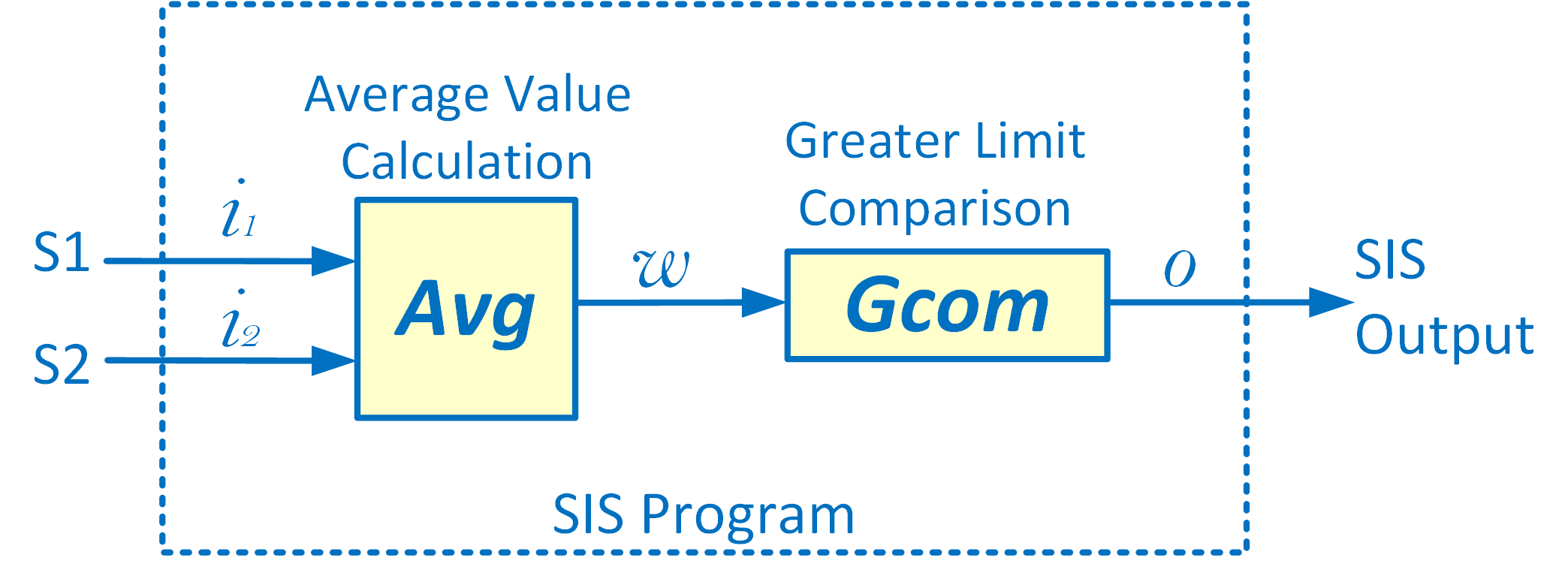}}
\quad
\subfloat[Program $T_{Or}$\label{SimpProgB}]{\includegraphics[width=0.43\columnwidth]{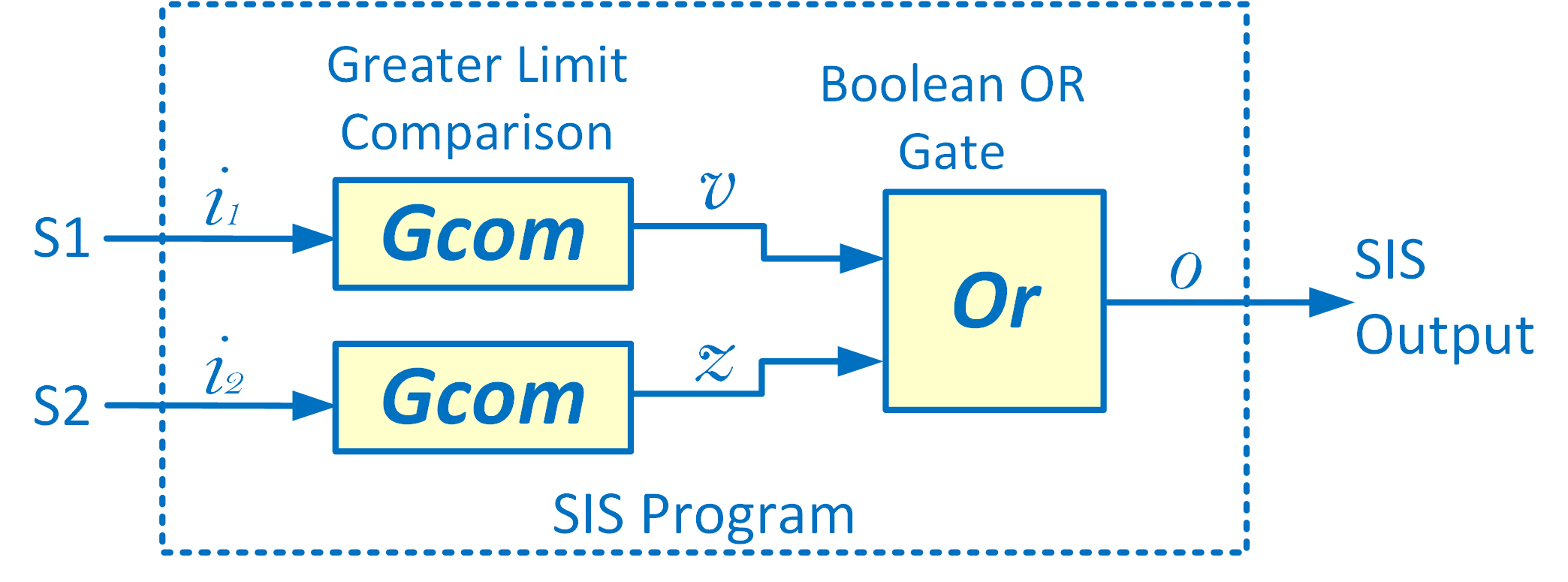}}\\}
\vspace{0.2 cm}
\begin{scriptsize}
For variables $i_1, i_2, w \in \mathbb{R}$ and $v, z, o \in \mathbb{B}$, and parameter $K \in \mathbb{R}$: $w=Avg(i_1, i_2)=(i_1+i_2)/2$, $o=Gcom_K(w)=(w>K)$ and $o=Or(v, z)=v \vee z$.
\end{scriptsize}
\caption{Two possible implementations for the Logic Solver in Fig.\ref{SimpFTAA}}
\label{SimpProg}
\end{figure}

In real world scenarios,  SIS programs are large and  complex. It is not unusual for a typical  SIS program to have hundreds of inputs like $i_1$ and $i_2$ and thousands of Function Blocks \cite{Ref_203} like $Avg$ and $Gcom_K$. Conducting a detailed analysis of program of such scales will be a real challenge for a human analyst, but it nonetheless plays a crucial part in producing accurate results. In such scenarios an automated method such as FMR can be of a great help.

FMR is a technique for enabling the details of SIS programs to be accurately reflected in the full safety analysis of a system.
The challenge we address is identifying the SIS input ``failure modes'' that cause a given SIS output failure mode by analysing the SIS program that reads those inputs and produces that output. The results can then be incorporated in an overall safety analysis of SIS. 

SIS programs are commonly developed in the form of Function Block Diagrams (FBD) \cite{Ref_203}. Fig.~\ref{SimpProg} showed two very simple examples of FBDs. An FBD consists of function blocks and their interconnections, which we label with variable names. In Fig.~\ref{SimpProgA}, $\Var{o}, \Var{w}, \Var{i}_{1}$ and $\Var{i}_{2}$ are the variables and $Avg$ and $Gcom_K$ are the function blocks. We will use this FBD as a worked example through this paper to demonstrate the FMR process.

The SIS program given at Fig. \ref{SimpProgA} is supposed to initiate a command to close the gas value when the pressure rises above a given threshold. In normal circumstances, when all inputs report correct measurements from the sensor readings, an output of $\True$ causes the correct shut down command to be delivered when the pressure is high. Suppose however that the inputs $i_1, i_2$ are incorrectly recording the pressure. These inaccuracies propagate through the program and lead to an $\False$ at the output, meaning that the SIS will not initiate the safety action required to prevent the hazardous event.

In simple terms, this is how FMR analyses such output deviations: from $\Var{o}$ being $\False$ by fault we can conclude that $\Var{w}$ must be less than the threshold limit set in $Gcom_K$: $(\Var{o}=\textsf{f}) \Imply (\Var{w}\leq K)$. Sentence $(\Var{w}\leq K)$ in turn implies that the average value of $\Var{i}_1$ and $\Var{i}_2$ must be less than the threshold limit: $(\Var{w}\leq K) \Imply ((\Var{i}_1+\Var{i}_2)/2\leq K)$. Assuming that this is due to an input fault, we can conclude that either input $\Var{i}_1$ must be reading lower than what it should, or input $\Var{i}_2$. Overall, we can conclude:
\begin{equation}\label{RevInfSt}
(\Var{o} \text{ being \False~by fault}) \Imply (\Var{i}_1 \text{ reads too low}) \vee (\Var{i}_2 \text{ reads too low}) 
\end{equation}

Notice that the actual values of inputs are not required, but only their categories in terms of whether they are ``too high'', or ``too low''. It turns out that we can take advantage of this abstraction to simplify the overall analysis. In the next section we describe a simple model of failures from which we derive an analysis that uses ``failure modes'' explicitly.\footnote{
Note that in Fig.~\ref{SimpProgB} the FMR analysis would produce a different result, i.e. $(\Var{o} \text{ being \False~by fault}) \Imply (\Var{i}_1 \text{ reads too low}) \land (\Var{i}_2 \text{ reads too low})$}

FMR completes the SIS safety analysis by incorporating the functionally most important part of the system – the program, and it does this by analysing the actual program rather than a synthesised model. The process is automated and thus it saves time and effort, and offers accuracy and certainty. The purpose of FMR is similar to fault tree analysis, but it adds rigour to the consideration of fault propagation in the SIS program.

\section{Modelling failures}\label{s1245}

In this section we formalise the ideas underlying the identification and analysis of potential failures for SIS programs. In particular the result of the analysis should be the identification of potential faults and, in addition, to categorise them in terms of their ``modes of failure''. This is an essential step in any safety engineering exercise.

In what follows we use well known constructions from relational-style modelling.  Our contribution is to apply those ideas in this new setting for SIS programs. Let ${\cal V}$ be an abstract state space; we use ${\mathbb P}{\cal X}$ for the power set over ${\cal X}$. A partition of a set ${\cal X}$ is a set of pairwise non-intersecting subsets in ${\mathbb P}{\cal X}$.

We begin with a simple abstract model for a generic SIS function. It is a function which takes inputs to outputs over an (abstract) type ${\cal V}$.

\begin{definition}
An abstract model for an SIS function is a function of type ${\cal V}\rightarrow {\cal V}$.
\end{definition}

An SIS function can be a function block (FB), a combination of FBs or the entire SIS program. As described above, the safety analyst can only access information about the safety status of the system through the SIS program variables.  The challenge is that this reported status (i.e.\ the sensor readings) might be inaccurate or wrong.  To model such  faults we need to keep track of the values recorded in the SIS program variables \emph{and} the value that \emph{should} have been reported. When these values are not the same we say that there is a fault. The next definition shows how to keep track of these faults within a particular SIS setting.

\begin{definition}\label{d1018}
Given an SIS function $f: {\cal V}\rightarrow {\cal V}$, a \emph{failure model}  is a function $\FunFail{f}: {\cal V}^2\rightarrow {\cal V}^2$ defined by
\[
\FunFail{f}(m, a) \Wide{\Defs} (f(m),f(a))~.
\]
For the pair $(m,a)\in {\cal V}^2$, the first component $m$ models the value reported by the SIS program variables, and the second component $a$ is the actual value that should be reported. We say that $(m, a)$ is a \emph{failure state} whenever $m\neq a$.
\footnote{In our abstract model we use a single type ${\cal V}$ for simplicity of presentation.}
\end{definition}

For example, in Fig.~\ref{SimpProgA} we model the
simple SIS program as a function $T_{Avg}$ of type $\mathbb {R}^2\rightarrow \mathbb{B}$, where the input (pair) corresponds to the readings of the variables $i_1, i_2$, and the output corresponds to the value of the output variable $o$.%
\footnote{Note here that we are distinguishing the types in the example.}
There are two possible output failure states wrt.\ $\FunFail{T_{Avg}} \in (\mathbb {R}^2)^2\rightarrow (\mathbb{B})^2$, and they are $(\True, \False)$ and $(\False, \True)$. 

Observe however from Def.~\ref{d1018} that the only way an output failure state can occur is if the corresponding input is also a failure state (since we are assuming that no additional failures are caused by the SIS program itself). Given a function $f$, we say that failure output state $(m', a')$ was \emph{caused by} input failure state $(m, a)$ if $\FunFail{f}(m, a)= (m', a')$.

In the case of  Fig.~\ref{SimpProgA}, the failure state $(\False, \True)$ can only be caused by input failure state $((m_1, m_2), (a_1, a_2))$ if either $m_1 < a_1$ or $m_2 < a_2$. Here the values $m_1, a_1$ correspond to the variable $i_1$ and $m_2, a_2$ correspond to the variable $i_2$ in the figure.  
In scenarios where e.g.\  $m_1< a_1$ there is always some reported value for $m_2$ such that the reported average $(m_1{+}m_2)/2$ is below the fixed threshold in $Gcom_K$, thus there exists a scenario satisfying the identified input constraints such that:
\[
\FunFail{T_{Avg}}((m_1, a_1), (m_2, a_2)) = (\False, \True)~.
\]

From this example we can see there are potentially infinitely many values for a failure  state $(m, a)$ whenever $m, a$ can take real values. Rather than a safety engineer needing to know these precise values, what is more relevant  is a report of the (usually) finite number of classes or \emph{modes} describing the kinds of failure.

\begin{definition}\label{d1611}
Given a set of states ${\cal V}{\times}{\cal V}$ wrt.\ a failure model, the \emph{failure modes} are defined by a partition ${\cal P}$ of ${\cal V}{\times}{\cal V}$. Each subset in ${\cal P}$ defines  a \emph{failure mode} (relative to ${\cal P}$).   Two states $(m, a)$ and $(m', a')$ satisfy the same failure mode if and only if they belong to the same partition subset of ${\cal P}$. 

Given a partition ${\cal P}$ defining a set of failure modes we define $\modes_{\cal P} : {\cal V}{\times}{\cal V} \rightarrow {\cal P}$ which maps failure states to their relevant failure mode (partition subset).
\end{definition}

Examples of failure modes are normally described by constraints on variables. For instance in Fig.~\ref{SimpProgA} the  failure modes for the initial failure state $((m_1, m_2), (a_1, a_2))$ are summarised by ``either $i_1$ is reading too low or $i_2$ is reading too low''.  In terms of Def.~\ref{d1611} 
this can be characterised by part of a partition that includes $\ell_1, \ell_2$ and $\ell$, where $\ell_1$ is the set of failure states such that $m_1<a_1 \land m_2\geq a_2$; $\ell_2$ is the set of failure states such that $m_1\geq a_1 \land m_2<a_2$ and $\ell$ is the set of failure states such that $m_1<a_1 \land m_2<a_2$.

Given an output failure mode, we would like to compute all initial failure modes that could cause that final failure mode. We say that an initial failure mode $\mathcal{e}$ (to an SIS function) \emph{causes} an output  failure mode $\mathcal{e}'$ (of an SIS function) if there exists a failure state satisfying $\mathcal{e}$ such that the output of the SIS function given that initial state satisfies $\mathcal{e}'$.  

For a given SIS function $f$, one way to do this is to compute all relevant failure states for $\langle f \rangle$, and then use $\modes_{\cal P}$ to interpret the failure modes for each failure state. Our first
observation is that, given a partition ${\cal P}$ defining the failure modes, 
we can simplify this procedure significantly by abstracting the behaviour of $f$ to act directly in terms of the failure modes rather than failure states.

\begin{definition}\label{d1639}
Let $f: {\cal V}\rightarrow {\cal V}$ be an SIS function, and 
${\cal P}$ be a partition of ${\cal V}^2$ defining the set of failure modes  as in Def.~\ref{d1611}. 

We define $\FunMode{f}{{\cal P}} ~: {\cal P} \rightarrow {\mathbb P}{\cal P}$ to be the \emph{failure mode abstraction} of $f$ as the (possibly nondeterministic) function satisfying the following constraint for any input $(m,a)\in {\cal V}^2$:
\[
{\modes_{\cal P}}\circ\FunFail{f}(m,a) \in \FunMode{f}{\cal P}\circ{\modes_{\cal P}}(m,a) ~.
\]
\end{definition}  

In  Fig.~\ref{SimpProgA}, where the initial failure modes are $\ell_1, \ell_2$ and $\ell$ explained above, and final failure modes are $\mathcal{f} = \{(\False, \True) \}$ and $\mathcal{t} = \{(\True, \False) \}$, we can see that $\FunMode{T_{Avg}}{\cal P}(\ell_1)$ contains $\mathcal{f}$, where we are writing ${\cal P}$ to represent the partition defined by all initial and final variables.\footnote{More precisely we would define failure modes separately on inputs and outputs, and indeed this is what happens in practice. To simplify the presentation however we assume that there is a single partition which serves to define failure modes on a single set, without distinguishing between inputs and outputs.}

We shall show below that there are a variety of functions that have well-defined failure mode abstractions. Our next task however, is to show that the abstraction defined by Def.~\ref{d1639} is compositional, i.e.\ the abstraction of  $f;g$ of SIS functions $f$ and $g$ can be computed from the composition of their abstractions. We recall the well-known Kleisli   lifting of set-valued \cite{Ref_212,Ref_213} functions as follows. We write the composition $f;g$ to mean first $f$ is executed, and  then $g$, or as functions the output from initial $s$ is $g(f(s))$. 

Let $\rho: \mathcal{T}\rightarrow {\mathbb P}\mathcal{T}$, define $\rho^\dagger: {\mathbb P}\mathcal{T}\rightarrow {\mathbb P}\mathcal{T}$
\begin{equation}\label{e1240}
 \rho^\dagger(K)  ~~~ \Defs ~~~ \bigcup_{k\in K}\rho(k)~.
\end{equation}

\begin{lemma}\label{l1106}
Let $f, g$ be SIS functions which have well-defined failure-mode abstractions as given by Def.~\ref{d1639}. The failure-mode abstraction for the composition $\FunMode{f;g}{\cal P}$ is equal to $\FunMode{g}{\cal P}^\dagger\circ \FunMode{f}{\cal P}$, where $\FunMode{g}{\cal P}^\dagger: {\mathbb P}{\cal P}\rightarrow  {\mathbb P}{\cal P}$ is the standard lifting set out at Eqn.~\ref{e1240} above.
\footnote{Recall that for simplicity we assume that the function modes ${\mathcal P}$ applies to both functions $f$ and $g$.}

\begin{proof} (Sketch)
We show, for any input $(m, a)$, that: 
\[
{\modes_{\cal P}}\circ\FunFail{g}\circ\FunFail{f}(m,a) \in \FunMode{g}{\cal P}^\dagger\circ\FunMode{f}{\cal P}\circ{\modes_{\cal P}}(m,a)~,
\]
and that all failure modes arise in this way.
  The result follows from Def.~\ref{d1639}, and standard manipulations of set-valued functions \cite{Ref_212,Ref_214}.
\end{proof}
\end{lemma}

The failure mode abstractions $\FunMode{f}{\cal P}$ enable a significant simplification in the identification of possible failures in an SIS program. For example we shall see that $\FunMode{T_{Avg}}{\cal P} = \FunMode{Gcom_K}{\cal P}^\dagger\circ\FunMode{Avg}{\cal P}$ for abstractions of the function blocks  $\FunMode{Gcom_K}{\cal P}$ and $\FunMode{Avg}{\cal P}$.

In general a safety analyst considers possible output failure modes and asks for the inputs that potentially cause them.  In some circumstances some failure modes can never be satisfied by any input, and are deemed \emph{unreachable}.  The analyst is thus able to concentrate on \emph{reachable} failure modes, defined next. 

\begin{definition}\label{d1227}
Given an SIS function $f$, and abstraction defined by Def.~\ref{d1639}. A failure mode $\mathcal{m}\in {\cal P}$ is \emph{reachable} (wrt.\ $f$) if there is some input failure state $(i, i')$ such that ${\modes_{\cal P}}\circ \FunFail{f}(i, i')= \mathcal{m}$.
\end{definition}

Failure Mode Reasoning is based on backwards calculational reasoning. We use a weak transformer to compute all input failure modes which can possibly cause a given output failure mode. This is similar to the dual transformer of dynamic logic \cite{Ref_223} and the conjugate transformer \cite{Ref_215} for the well-known guarded command language \cite{Ref_216}.

\begin{definition}\label{d1230}
Given SIS function%
\footnote{We do not treat non-termination nor partial functions.}
$f$, we define the \emph{inverse failure transformer} ${\RevFunMode{f}{\cal P}}: \mathbb{P}{\cal P}\rightarrow \mathbb{P}{\cal P}$ as 
\begin{equation}\label{e1601}
{\RevFunMode{f}{\cal P}}(K) ~~\Defs ~~ \{\mathcal{k} ~| ~  \FunMode{f}{\cal P}(\mathcal{k}) \cap K\neq \phi \}~.
\end{equation}
\end{definition}

Def.~\ref{d1230} satisfies two properties. The first is that any initial failure modes computed from final failure modes
are the ones that could cause the selected final failure modes.
 The second is that inverse failure transformers compute all initial failure modes from final reachable failure modes. The next two definitions formalise these properties.

\begin{definition}\label{d1619}
Given SIS function $f$, we say an inverse failure transformer $t$ is \emph{sound} wrt.\ $f$ if all 
$\mathcal{k} \in t(\mathcal{K})$ implies $\FunMode{f}{\cal P}(\mathcal{k}) \cap \mathcal{K} \neq \phi$.
\end{definition}

\begin{definition}\label{d1531}
Given  SIS function $f$, we say an inverse failure transformer $t$ is \emph{complete} if for any set of reachable failure modes ${\cal F}$ and (initial) failure modes ${\cal I}$, we have the following:
\begin{equation}\label{e1610}
{\cal I} \subseteq t({\cal F})~~~~ \Leftrightarrow ~~~~ (\forall \mathcal{i}\in \mathcal{I} \cdot \FunMode{f}{\cal P}({\cal i}) \cap {\cal F} \neq \phi)~.
\end{equation}
\end{definition}

Observe that given failure modes $\mathcal{m}$  and $\mathcal{m}'$ such that $\mathcal{m} \in t\{\mathcal{m}'\}$, then $\mathcal{m}'$ is reachable if there is some $(i, i')$  such that 
${\modes}_{\cal P}(i, i')= \mathcal{m}$. In general the safety engineer is not concerned with ``unrealistic''  failure modes in the sense that  no corresponding scenario comprised of failure states can be constructed.

It is clear from Def.~\ref{d1230} that $\RevFunMode{f}{\cal P}$ is a sound and complete transformer relative to $f$.
The definition of completeness is important because it means, for the safety engineer, that all potential failure modes are accounted for by the abstraction. The next lemma records the fact that soundness and completeness is conserved by function composition. 

\begin{lemma}\label{l1231}
Let $f, g$ be SIS functions, and let $\mathcal{P}$ determine the failure modes so that $\RevFunMode{f}{\mathcal{p}}$ and $\RevFunMode{g}{\mathcal{p}}$ are sound and complete transformers.  Then their composition $\RevFunMode{f}{\mathcal{p}}\circ \RevFunMode{g}{\mathcal{p}}$ is also sound and complete for the composition SIS function $f;g$.
\begin{proof}
Follows from Def.~\ref{d1230} and standard facts about functions and their transformers \cite{Ref_212,Ref_215}. 
\end{proof}
\end{lemma}

In this section we have set out a formal methods treatment of failure modes for SIFs in SIS programs. We have demonstrated a simple model for failures and shown how this ``application-oriented" approach supports a rigorous analysis of failure modes and how they are propagated in SIS programs. In the following sections we show how this can be used to justify the use of standard backwards-reasoning to compute all input failure modes that cause reachable failure modes.

\section{Failure mode reasoning}\label{fmrsec}

In this section we show how to apply the failures model introduced in Section\ref{s1245} to the typical safety analysis.

Recall $T_{Avg}$ defined in Fig.~\ref{SimpProgA}. In this example, the failure modes of interest relate to whether the readings of the various sensors accurately record the physical environment or not, and when they do not, which combinations of deviant readings have the potential to result in a hazard. 

The safety analysis begins with the identification of hazardous outputs: these are outputs from the SIS program which would directly cause a hazard if it is not correct, in the sense that it deviates from the ``true'' result which would have been output had all the sensors accurately recorded the status of the plant.

For simplicity we assume that all readings are real-valued, thus we identify ``True'' with ``1'' and ``False'' with ``0''. Following Def.~\ref{d1611} we set ${\cal V} = \mathbb{R}$ and identify a partition on $\mathbb{R}{\times}\mathbb{R}$ given as follows.

\begin{definition}\label{d1515}
Define the \emph{failures partition} as follows. Let $\mathcal{h}, \mathcal{l}, \mathcal{m}$ respectively partition $\mathbb{R}{\times}\mathbb{R}$ defined by:
\[
(r, r')\in \mathcal{h} ~~\textit{iff}~~ r{>}r'~~;~~
(r, r')\in \mathcal{m} ~~\textit{iff}~~ r{=}r'~~;~~
(r, r')\in \mathcal{l} ~~\textit{iff}~~ r{<}r'~.
\]
\end{definition}

Here we have identified the common failure modes ``reading too high'', corresponding to $\mathcal{h}$ and ``reading too low'' corresponding to $\mathcal{l}$. We have also included ``reading correct'' corresponding to $\mathcal{m}$ which is not strictly speaking a ``failure'', but is useful in the formal analysis. From our gas pressure example, the situation where the input recorded on $\Var{i}_1$ is lower than the real pressure in the pipe is modelled by pairs of values that lie in $\mathcal{l}$.

Safety engineers want to know the input failure modes that ``cause" particular reachable output failures. Def.~\ref{d1531} and Lem.~\ref{l1231}  above support a standard backwards reasoning method on failure modes directly.

For each variable $s$ in an SIS program we use $\FVal{s}$ for a corresponding variable taking failure modes for values, which in this case is  
 $\{\mathcal{h}, \mathcal{l}, \mathcal{m}\}$.

\begin{definition}\label{d1334}
Given an SIS function $f$ and a partition ${\cal P}$ defining the failure modes. A \emph{failure triple} is written
\begin{equation}\label{e1340}
\Scenario{\FVal{s}\in \mathcal{A}}{f}{\FVal{s}'\in \mathcal{A}'}~,
\end{equation}
where $\mathcal{A}, \mathcal{A}'\subseteq \{\mathcal{h}, \mathcal{l}, \mathcal{m}\}$. The triple Eqn.~\ref{e1340} is \emph{valid} if, for each failure mode $\mathcal{e} \in \mathcal{A}$ there exists $(m, m')$ such that ${\modes}_{\cal P}(m, m') = \mathcal{e}$ and ${\modes}_{\cal P}(\FunFail{f}(m, m'))\cap\mathcal{A}'\neq \phi$.

Note that as a special case where $\mathcal{A}$ is a singleton set $\{\mathcal{a} \}$ we write ``$\FVal{s}=\mathcal{a}$" rather than  ``$\FVal{s}\in \mathcal{A}$".
\end{definition}

Def.~\ref{e1340} is reminiscent of a standard Hoare Triple for failure modes, however a failure triple is based on Def.~\ref{d1230}.  More importantly Def.~\ref{e1340} corresponds with the scenarios relevant for the assessment of failures. Whenever $f$ corresponds to an SIS function for example, the valid triple given by Eqn.~\ref{e1340} means that the initial failure mode corresponding to $\mathcal{a}$ \emph{causes} the final failure mode $\mathcal{a}'$. This effectively enables the identification of failure mode propagation, summarised in the next result.

\begin{theorem}\label{t1554}
Let $f$ be an SIS function and ${\cal P}$ define the relevant failure modes. Let $\mathcal{a}'$ be a reachable final failure mode wrt.\ $f$. Then for all $\mathcal{a}\in {\RevFunMode{f}{\cal P}}\{\mathcal{a}'\}$
\[
\Scenario{\FVal{s}=\mathcal{a}}{f}{\FVal{s}'=\mathcal{a}'}
\]
is a valid failure triple.
\begin{proof}
Definition of ${\RevFunMode{f}{\cal P}}$, Def.~\ref{d1230}.
\end{proof}
\end{theorem}

\paragraph{Backwards reasoning for failure modes:}
As mentioned above we can use Thm.~\ref{t1554} to compute the failure modes that are the cause of a given reachable final failure mode. 
A complex SIS program determining a SIF typically comprises multiple function blocks with clearly defined ``input'' variables and ``output'' variables, where
the outputs are determined by the values on the inputs. The architecture of the SIS program is then equated with a composition of a series of
function blocks. Now that we have a formal description in terms of failure triples, we are able to use the standard composition rule: 

\[
\begin{array}{l}
\Scenario{\FVal{s}=\mathcal{a}}{f_1}{\FVal{s_1}=\mathcal{b}} ~\land~ \Scenario{\FVal{s_1}=\mathcal{b}}{f_2}{\FVal{s}'=\mathcal{a}'}\\ ~~\Rightarrow~~
\Scenario{\FVal{s}=\mathcal{a}}{f_1;f_2}{\FVal{s}'=\mathcal{a}'}  ~.
\end{array}
\]

From this we can now deduce failure triples of a complex SIS program by reasoning about failure triples for component function blocks. We illustrate this for $Avg$ and $Gcom_K$ in the next section.

\section{Individual function blocks}\label{rfbssec}
A typical SIS program library, from which function blocks (FBs) are chosen, may include 100 types of FBs \cite{Ref_219}. For each FB the relationships between FB input failure modes and FB output failure modes, can be summarised in a Failure Mode Block (FMB). An FMB is proposed based on the well-defined function of its corresponding FB. In this section we will propose FMBs for SIS functions $Avg$ and $Gcom_K$, which we used in our gas pressure example, and we will prove the soundness and completeness of the proposed FMBs. More sample FMBs are proposed and proven in the Appendix.

The $Avg$ function block takes two inputs and computes the average. The relevant output failures therefore are whether the output reads too high or too low. The abstraction for failure modes is given below.

\begin{definition}\label{Avg}
Let $Avg$ be the function defined by: $Avg(i_1, i_2) \Defs (i_1{+}i_2)/2$. Its associated FMB, FAvg, is defined as follows:
\[
\begin{array}{ccc}
\{ \FVal{i}_1={\mathcal{h}}\lor \FVal{i}_2={h} \}&~~Avg~~& 
\{\FVal{o}= \mathcal{h}\}\\
\{ \FVal{i}_1=\mathcal{l}\lor \FVal{i}_2=\mathcal{l} \}&~~Avg~~& 
\{\FVal{o}= \mathcal{l}\}\\\\
\left\{\begin{array}{l}
 \FVal{i}_1=\mathcal{h} \land\FVal{i}_2=\mathcal{l} \lor\\
\FVal{i}_1=\mathcal{l} \land\FVal{i}_2=\mathcal{h} \lor \\
\FVal{i}_1=\mathcal{m} \land\FVal{i}_2=\mathcal{m} 
\end{array}\right\}
&~~ Avg~~& 
\{\FVal{o}= \mathcal{m}\}
\end{array}
\]
\end{definition}

Def.~\ref{Avg} tells us that if the output reads too high, then it must be because one of the two inputs also reads too high. Similarly, if the output reads too low then it can only be because one of the two inputs reads too low. On the other hand the output can deliver an accurate result for scenarios where one input reads too high and the other reads too low.  At the qualitative level of abstraction, however, all of these possibilities must be accounted for.

$Gcom_K$ is another typical function block which compares the input with a given threshold and reports whether the input meets the given threshold.

\begin{definition}\label{Gcom}
Let $Gcom_K$ be the function defined by: $Gcom_K(i) \Defs (i > K)$. Its associated FMB, FGcom, is defined as follows:
\[
\begin{array}{ccc}
\{ \FVal{i}=\mathcal{h}\}&~~Gcom_K~~& 
\{\FVal{o}= \mathcal{t}\}\\
\{ \FVal{i}=\mathcal{l} \}&~~Gcom_K~~& 
\{\FVal{o}= \mathcal{f}\}\\
\{ \FVal{i}=\mathcal{l} \lor \FVal{i}=\mathcal{m} \lor \FVal{i}=\mathcal{h}\} 
&~~ Gcom_K~~& 
\{\FVal{o}= \mathcal{m}\}
\end{array}
\]
\end{definition}

Def.~\ref{Gcom} tells us that the output reading $\False$ when it should read $\True$ can only happen when the input is delivering a lower value than it should, and similarly the output reading $\True$ when it should read $\False$ can only happen when  the input reading is falsely reporting a high value. Notice that this definition is actually independent of $K$, which is why $K$ is suppressed in the FMB model. 

The following theorem confirms that Def.~\ref{Avg} and Def.~\ref{Gcom} are sound and complete in respect of their operational definitions.

\begin{theorem}\label{CorrComp}
The $FAvg$ and $FGcom$ models Definitions Def.~\ref{Avg} and Def.~\ref{Gcom} are the sound and complete failure models of $Avg$ and $Gcom_K$ (for all real-valued $K$).
\end{theorem}

\begin{proof}
Individual FMBs can be proven by using truth-tables. All possible combinations of faults at the inputs and outputs of a corresponding FB can be defined, based on which the soundness and completeness conditions can be examined. Detailed proof is given in Appendix.
\end{proof}

\section{FMR in practice}\label{rflsec}

The FMR process consists of four main stages: composition, substitution, simplification and calculation. In the composition stage, FMBs and failure mode variables are defined and connected in accordance with the SIS program. The model for our example SIS program (Fig.~\ref{SimpProgA}) will include two FMBs: $FAvg$ and $FGcom$. Similarly, variables $\Var{o}, \Var{w}, \Var{i}_{1}$ and $\Var{i}_{2}$ in SIS program will have their own corresponding failure mode variables $\FVal{o}, \FVal{w}, \FVal{i}_{1}$ and $\FVal{i}_{2}$ in the model.

The reasoning process begins at the last FB, i.e. the one that produces the SIS output. In our gas pressure example, the given output fault is $\FVal{o}=\mathcal{f}$. Taking into account the function of $Gcom_K$ from Def.~\ref{Gcom}, we can say:
\begin{equation}\label{GPS1}
\{\FVal{ w }=\mathcal{l} \}~~Gcom_K~~ 
\{\FVal{o}= \mathcal{f}\}\
\end{equation}

Statement (\ref{GPS1}) suggests that output $\Var{o}$ being $\False$ by fault implies that the input to the greater comparison FB, $\Var{ w }$, is reading lower than what it should. 

The reasoning process continues through the SIS program until all the conclusion parts of the implication statements include no more intermediate variables. In our example, the next FB is $Avg$. Considering the function of $Avg$, if the fault $\FVal{w}=\mathcal{l}$ occurs at its output, we can conclude that from Def.~\ref{Avg}:
\begin{equation}\label{GPS2}
\{ \FVal{i}_1=\mathcal{l}\lor \FVal{i}_2=\mathcal{l} \}~~Avg~~ 
\{\FVal{ w }= \mathcal{l}\}
\end{equation}

This statement suggests that if the reported value at $\Var{w}$ is lower than its intended value, then either input $\Var{i}_{1}$ or $\Var{i}_{2}$ may be reading lower. The reasoning sequence terminates here as the left hand side of (\ref{GPS2}) only includes SIS inputs.

In the second stage of FMR we use the logical composition rules to eliminate intermediate variables in order to reduce the set of FB failure reasons to only one relation that links SIS inputs to its outputs.  In our example, the only internal variable is $\FVal{w}$. By substituting (\ref{GPS2}) in (\ref{GPS1}) we can conclude: 
\begin{equation}\label{GPS3}
\{ \FVal{i}_1=\mathcal{l}\lor \FVal{i}_2=\mathcal{l} \}~~Avg;Gcom_K~~ 
\{\FVal{o}= \mathcal{f}\}
\end{equation} which is very similar to the result (\ref{RevInfSt}) of our earlier informal description of FMR.

The third stage of FMR is simplification, where we use standard rules of propositional logic \cite{Ref_226} to simplify (\ref{GPS3}) and create the FMR short list of failure triples. As (\ref{GPS3}) is already minimal, we can easily see that our short list of faults comprises  $\FVal{i}_1=\mathcal{l}$ and  $\FVal{i}_2=\mathcal{l}$.

Having the input failure modes identified, we can implement the last stage of analysis, calculation, in which we would assign probability values to individual failure events and calculate the overall probability of failure. We skip this stage for this simple example. A comprehensive safety analysis for a realistic case study is described in other work \cite{Ref_200}. 

In a more recent project \cite{Ref_225} we examined a larger case study where we integrated FMR with other model-base analysis methods HiP-HOPS \cite{Ref_140} and CFT \cite{Ref_143}. We demonstrated that not only is FMR able to handle larger examples with precision, but its output can also be of value to other safety analysis tools that are designed to model generic systems but not programs. The process we examined in this case study is briefly shown in Fig. \ref{Fig_SIF_Config}: a SIS that protects a gas-fired industrial boiler against high level of water. The SIS program in this example consists of over 2170 function blocks. With close to 100 inputs and over 25 outputs, the SIS performs a total of 34 safety functions (SIFs). The SIS program in this project was developed in FBD and saved in XML format.
 
\begin{figure}[!ht]
\includegraphics[scale=0.4]{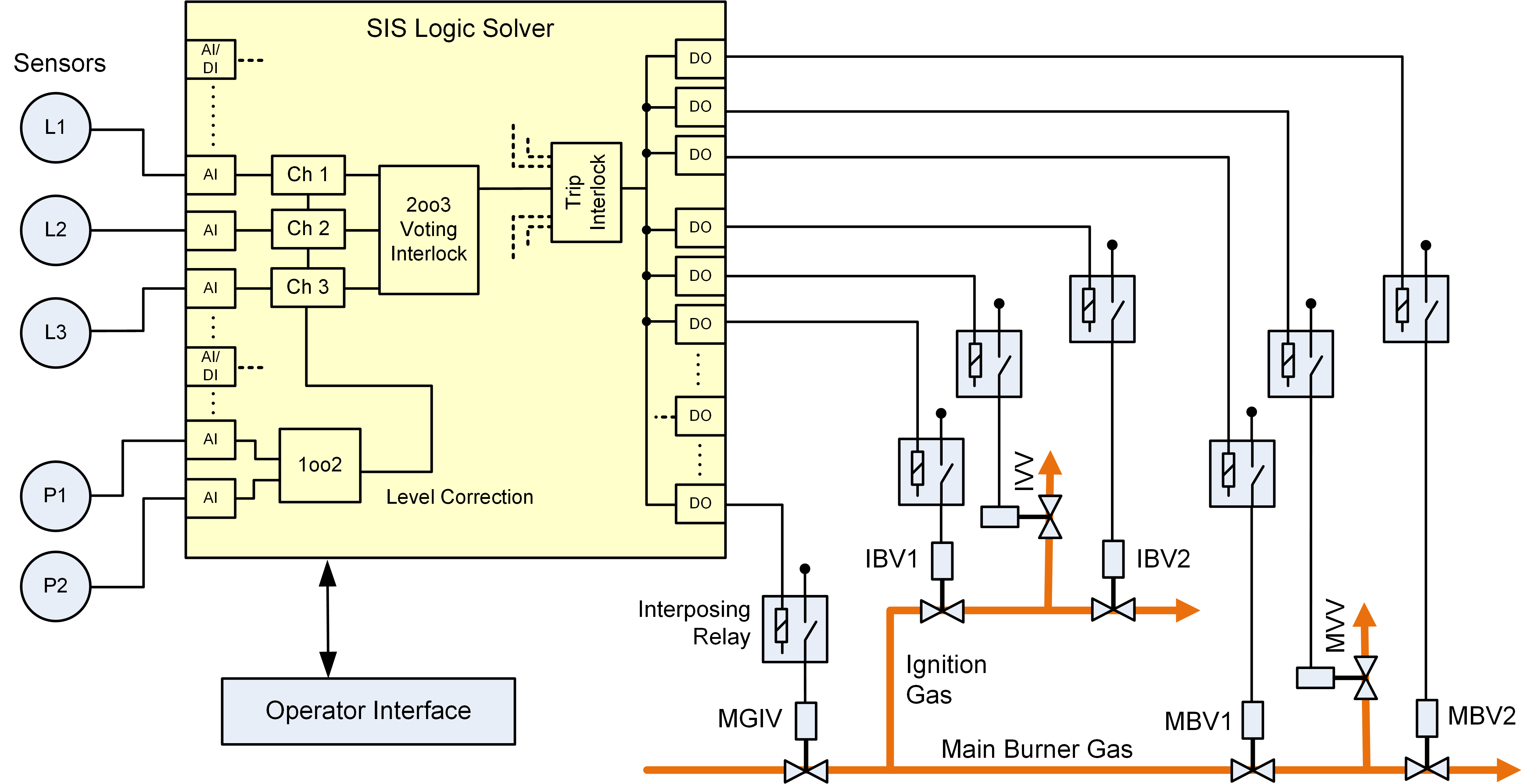}
\centering
\caption{SIS configuration}
\label{Fig_SIF_Config}
\end{figure}

The FMR analysis produced two short lists of failure modes, one for Dangerous Undetected (DU) failure and one for Spurious Trip (ST). The lists included a total of 39 failure scenarios. In the quantitative stage the failure data of SIS inputs were entered and the aggregated probability measures for DU and ST failures were calculated.

Provided that failure data are readily available, the whole analysis process for an SIS of this scale takes less than an hour to complete, using the experimental system incorporating FMR analysis \cite{Ref_225}. Conducting similar analysis by hand would take days. To visualise the extent of work, consider manual implementation of a fault tree with around 3500 gates. Even if the analyst is prepared for such a challenge, the implemented model, and thus its outcome, will be prone to human error. In comparison, FMR is fast, accurate, consistent, and reliable.

\section{Discussion and related works}\label{discsec}

Reasoning about faults is not a new research topic. Diagnostics based on systematic inference was extensively studied in the 1980's. Some of the frequently cited articles include \cite{Ref_228,Ref_233,Ref_227}. Generally speaking, the studies were aimed at answering one question: given an observed deviation at the output of a system, how can we identify the (potentially) faulty components by reasoning based on the knowledge of system structure and/or system function? Logic circuits, in particular, would make an interesting application as they typically consist of complex yet well-defined, logical structures. Unlike inference-based diagnostics, FMR is primarily designed to target probable \emph{input faults}, rather than faulty system components. Input faults are external to the system and do not represent system failure scenarios.

FMR uses abstraction techniques, which is also a well-established area, particularly in formal methods \cite{Ref_235}. One may find similarities between the abstraction in FMR and that of Qualitative Reasoning (QR), where quantitative aspects of physical systems are replaced with qualitative representations \cite{Ref_238,Ref_232}. It should be noted however that QR is a method for expressing physical entities, and with an application in AI; whereas FMR is a technique for reasoning about failures, and (at least, currently) focused on conventional safety systems.

FMR is in some respects similar to FTA. Both methods look at the root causes that can result in a given top event. Parts of the computation techniques are similar between the two methods as well. However, FMR and FTA are different in some conceptual respects. FTA is a generic method that can be applied to any fault in any type of system, whereas FMR is specifically designed for analysing SIS programs. FTA computes a Boolean state of failure-success, but FMR computes multiple failure modes. The top event in FTA is a single event, but the program output in FMR can be an array of variables. The main question FMR tries to answer is that: given an abstracted state of output and given the function that produces it, what are the possible (abstracted) states of inputs to that function. This is obviously different to FTA in which we ``know'' the failure behaviour of a system and we build a model (fault tree) to summarise our understanding. FTA relies on the knowledge and skills of the analyst whereas FMR extracts information directly from the system. In a general term, FTA is a failure modeling method while FMR is a mode calculation method. 

FTA was first introduced in 1961 to study a missile launch control system. In almost six decades, many extensions and variations of the method have been introduced to solve other types of problems. Useful surveys are conducted on FTA and its extensions in recent years \cite{Ref_145,Ref_138}. Thanks to the growing capabilities of today’s technology, attention has shifted towards modularity and automatic synthesis of fault trees, which can greatly assist with solving complex problems at less effort. Various model-based dependability analysis methods have been developed, such as HiP-HOPS \cite{Ref_140,Ref_142}, AADL \cite{Ref_150} and AltaRica \cite{Ref_151}, which use FTA as their primary means and automate the synthesis process to some degrees. More recently, the concept of contract-based design has also been used for automatic generation of hierarchical fault trees from formal models \cite{Ref_220}.

The common concept in automatic hierarchical synthesis of fault trees is that if we have the typical definition of component fault trees, we can synthesise the system level fault tree by interconnecting the smaller fault trees of components. At a conceptual level, this idea is utilised by FMR too; however, the \textit{components} in FMR are the FBs, as opposed to the other methods that analyse physical systems. Also, while FMR uses the actual SIS program for its analysis, the other methods rely on separate models or specifications in order to generate fault trees. The actual running program in SIS is always the most accurate, detailed, and specific source of information on the behaviour of system, and having that FMR does not require any additional models.

Model checking has been used in SIS related applications too (see \cite{Ref_131,Ref_136} as examples). In model checking a formal specification of system (model) is checked against a formal specification of requirements. Such methods focus on verifying the program against the requirements, as opposed to FMR which aims to identify failure modes.

Satisfiability Modulo Theories (SMT) is about determining whether a first order formula is satisfiable with respect to some logical theory \cite{Ref_123,Ref_221}. SMT solvers are used in various applications in the field of computer and formal verification. With respect to FMR, SMT can potentially help with determining the SIS input values that can result in a given output value. While this makes a potential area for further research; our experiments so far indicate that any SMT analysis will require post-processing in order to transform the results into failure modes.

\section{Conclusion}\label{concsec}

In this paper we have shown how techniques from traditional formal methods can be brought to bear on a challenging problem in safety engineering: that of determining with precision how faults arising from incorrect sensor readings propagate through complex SIS programs.
Within  the safety engineering discipline, FMR is a novel way to analyse failure modes in Safety Instrumented Systems. Future work will include more complex constructs for function blocks, including looping, timing and probabilistic analysis. Moreover, we are working on implementing FMR for identify systematic failures in SIS programs, where the input to the program is correct but the output is faulty due to a pre-existing error in program.

\bibliographystyle{plain}
\bibliography{References}

\pagebreak
\appendix

\newcommand{\hjVar}[1] {#1} 
\newcommand{\hjFVar}[1] {\Hat{#1}} 
\newcommand{\hjFVal}[1]{\mathcal{#1}}
\newcommand{\hjValRl}[1]{
    \IfEqCase{#1}{
        {x}{a}
        {y}{a'}
        {#1}{#1}
    }
}
\newcommand{\hjValRp}[1]{
    \IfEqCase{#1}{
        {x}{m}
        {y}{m'}
        {#1}{\tilde{#1}}
    }
}

\newcommand{\TE}[1]{#1} 
\newcolumntype{M}[1]{>{\centering\arraybackslash}m{#1}}
\newcolumntype{N}{@{}m{0pt}@{}}

\section{Truth tables for individual FBs}

All possible failure scenario related to $Avg$ are summarised in Table \ref{STAvg}. In this Table, the 2\textsuperscript{nd} and 3\textsuperscript{rd} columns indicate the relationships between reported and intended values at inputs $\hjVar{x}_1$ and $\hjVar{x}_2$ and the 4\textsuperscript{th} column shows the type of fault at the output $\hjVar{y}$, caused by the inputs $\hjVar{x}_1$ and $\hjVar{x}_2$. A question mark indicates that the relation between $\hjValRp{y}$ and $\hjValRl{y}$ cannot be determined; i.e., all faults are possible. The last three columns in Table \ref{STAvg} translate the 2\textsuperscript{nd}, 3\textsuperscript{rd} and 4\textsuperscript{th} columns into failure modes, as used in the FMR modeling.

\begin{longtable}[!h]{|M{0.7cm}|M{2.4cm}|M{2.4cm}|M{2.4cm}|M{0.9cm}|M{0.9cm}|M{0.9cm}|}
\hline & & & & & & \\[-0.9em]
\multicolumn{1}{|c|}{\TE{no.}}&\multicolumn{1}{c|}{\TE{$(x_1 = m_1, a_1)$}}&\multicolumn{1}{c|}{\TE{$(x_2 = m_2, a_2)$}}&\multicolumn{1}{c|}{\TE{$(y = m', a')$}}&\multicolumn{1}{c|}{\TE{$\hjFVar{x}_1$}}&\multicolumn{1}{c|}{\TE{$\hjFVar{x}_2$}}&\multicolumn{1}{c|}{\TE{$\hjFVar{y}$}}\\ 
& & & & & & \\[-1em]
\specialrule{0.2em}{0.0em}{0.0em}  & & & & & & \\[-1em]
\TE{1} & \TE{$m_1<a_1$} & \TE{$m_2<a_2$} & \TE{$m'<a'$} & \TE{$\hjFVal{l}$} & \TE{$\hjFVal{l}$} & \TE{$\hjFVal{l}$} \\& & & & & & \\[-1em] \hline & & & & & & \\[-0.95em] 
\TE{2} & \TE{$m_1<a_1$} & \TE{$m_2=a_2$} & \TE{$m'<a'$} & \TE{$\hjFVal{l}$} & \TE{$\hjFVal{m}$} & \TE{$\hjFVal{l}$} \\& & & & & & \\[-1em] \hline & & & & & & \\[-0.95em] 
\TE{3} & \TE{$m_1<a_1$} & \TE{$m_2>a_2$} & \TE{$m'~?~a'$} & \TE{$\hjFVal{l}$} & \TE{$\hjFVal{h}$} & \TE{$\hjFVal{a}$} \\& & & & & & \\[-1em] \hline & & & & & & \\[-0.95em] 
\TE{4} & \TE{$m_1=a_1$} & \TE{$m_2<a_2$} & \TE{$m'<a'$} & \TE{$\hjFVal{m}$} & \TE{$\hjFVal{l}$} & \TE{$\hjFVal{l}$} \\& & & & & & \\[-1em] \hline & & & & & & \\[-0.95em] 
\TE{5} & \TE{$m_1=a_1$} & \TE{$m_2=a_2$} & \TE{$m'=a'$} & \TE{$\hjFVal{m}$} & \TE{$\hjFVal{m}$} & \TE{$\hjFVal{m}$} \\& & & & & & \\[-1em] \hline & & & & & & \\[-0.95em] 
\TE{6} & \TE{$m_1=a_1$} & \TE{$m_2>a_2$} & \TE{$m'>a'$} & \TE{$\hjFVal{m}$} & \TE{$\hjFVal{h}$} & \TE{$\hjFVal{h}$} \\& & & & & & \\[-1em] \hline & & & & & & \\[-0.95em] 
\TE{7} & \TE{$m_1>a_1$} & \TE{$m_2<a_2$} & \TE{$m'~?~a'$} & \TE{$\hjFVal{h}$} & \TE{$\hjFVal{l}$} & \TE{$\hjFVal{a}$} \\& & & & & & \\[-1em] \hline & & & & & & \\[-0.95em] 
\TE{8} & \TE{$m_1>a_1$} & \TE{$m_2=a_2$} & \TE{$m'>a'$} & \TE{$\hjFVal{h}$} & \TE{$\hjFVal{m}$} & \TE{$\hjFVal{h}$} \\& & & & & & \\[-1em] \hline & & & & & & \\[-0.95em] 
\TE{9} & \TE{$m_1>a_1$} & \TE{$m_2>a_2$} & \TE{$m'>a'$} & \TE{$\hjFVal{h}$} & \TE{$\hjFVal{h}$} & \TE{$\hjFVal{h}$} 
\\[0.15em]\hline 
\caption {Truth-table for $Avg$}\label{STAvg}
\end{longtable}

To use Table \ref{STAvg} for FB $Avg$ defined by Def. \ref{Avg}, recall that $\FVal{x}_1 \equiv \FVal{i}_1$, $\FVal{x}_2 \equiv \FVal{i}_2$ and $\FVal{y} \equiv \FVal{o}$. To prove Theorem \ref{CorrComp}, all we need to do is to group the combinations of $\FVal{x}_1$ and $\FVal{x}_2$ that correspond to $\mathcal{l}$, $\mathcal{h}$, and $\mathcal{m}$ in $\FVal{y}$ column. It is evident from Table \ref{STAvg} that rows 1-4 compose $\FVal{y} = \mathcal{l}$, rows 6-9 compose $\FVal{y} = \mathcal{h}$, and rows 3, 5 and 7 compose $\FVal{y} = \mathcal{m}$. Compare these combinations with the ones given in Def. \ref{Avg}.

Likewise, Table \ref{STGcom} can be used to prove Theorem \ref{CorrComp} for FB $Gcom_K$ as defined by Def. \ref{Gcom}. 

\begin{longtable}[!h]{|M{0.7cm}|M{2.4cm}|M{0.9cm}|M{0.9cm}|M{0.9cm}|M{0.9cm}|}
\hline & & & & & \\[-0.9em]
\multicolumn{1}{|c|}{\TE{no.}}&\multicolumn{1}{c|}{\TE{$(x= m, a)$}}&\multicolumn{1}{c|}{\TE{$m'$}}&\multicolumn{1}{c|}{\TE{$a'$}}&\multicolumn{1}{c|}{\TE{$\hjFVar{x}$}}&\multicolumn{1}{c|}{\TE{$\hjFVar{y}$}}\\ 
& & & & & \\[-1em]
\specialrule{0.2em}{0.0em}{0.0em}  & & & & & \\[-1em]
\TE{1} & \TE{$K < m < a$} & \TE{$\True$} & \TE{$\True$} & \TE{$\hjFVal{l}$} & \TE{$\hjFVal{m}$}\\& & & & & \\[-1em] \hline & & & & & \\[-0.95em] 
\TE{2} & \TE{$K < m = a$} & \TE{$\True$} & \TE{$\True$} & \TE{$\hjFVal{m}$} & \TE{$\hjFVal{m}$}\\& & & & & \\[-1em] \hline & & & & & \\[-0.95em] 
\TE{3} & \TE{$K < a < m$} & \TE{$\True$} & \TE{$\True$} & \TE{$\hjFVal{h}$} & \TE{$\hjFVal{m}$}\\& & & & & \\[-1em] \hline & & & & & \\[-0.95em] 
\TE{4} & \TE{$m \leq K < a$} & \TE{$\False$} & \TE{$\True$} & \TE{$\hjFVal{l}$} & \TE{$\hjFVal{f}$}\\& & & & & \\[-1em] \hline & & & & & \\[-0.95em] 
\TE{5} & \TE{$a \leq K < m$} & \TE{$\True$} & \TE{$\False$} & \TE{$\hjFVal{h}$} & \TE{$\hjFVal{t}$}\\& & & & & \\[-1em] \hline & & & & & \\[-0.95em] 
\TE{6} & \TE{$m < a \leq K$} & \TE{$\False$} & \TE{$\False$} & \TE{$\hjFVal{l}$} & \TE{$\hjFVal{m}$}\\& & & & & \\[-1em] \hline & & & & & \\[-0.95em] 
\TE{7} & \TE{$m = a \leq K$} & \TE{$\False$} & \TE{$\False$} & \TE{$\hjFVal{m}$} & \TE{$\hjFVal{m}$}\\& & & & & \\[-1em] \hline & & & & & \\[-0.95em] 
\TE{8} & \TE{$a < m \leq K$} & \TE{$\False$} & \TE{$\False$} & \TE{$\hjFVal{h}$} & \TE{$\hjFVal{m}$} 
\\[0.15em]\hline 
\caption {Truth-table for $Gcom_K$}\label{STGcom}
\end{longtable}

In the remaining part of this Appendix we propose FMBs for function blocks $Add$, $Sub$, $Abs$, $Lcom_K$, $Not$, $And$ and $Or$, and we present truth-tables that can be used to prove them.

\begin{definition}\label{Add}
Let $Add$ be the function defined by: $Add(x_1, x_2) \Defs x_1{+}x_2$. Its associated FMB, FAdd, is defined as follows:
\[
\begin{array}{ccc}
\{ \FVal{x}_1={\mathcal{h}}\lor \FVal{x}_2={h} \}&~~Add~~& 
\{\FVal{y}= \mathcal{h}\}\\
\{ \FVal{x}_1=\mathcal{l}\lor \FVal{x}_2=\mathcal{l} \}&~~Add~~& 
\{\FVal{y}= \mathcal{l}\}\\\\
\left\{\begin{array}{l}
 \FVal{x}_1=\mathcal{h} \land\FVal{x}_2=\mathcal{l} \lor\\
\FVal{x}_1=\mathcal{l} \land\FVal{x}_2=\mathcal{h} \lor \\
\FVal{x}_1=\mathcal{m} \land\FVal{x}_2=\mathcal{m} 
\end{array}\right\}
&~~ Add~~& 
\{\FVal{y}= \mathcal{m}\}
\end{array}
\]
\end{definition}

\begin{longtable}[!h]{|M{0.7cm}|M{2.4cm}|M{2.4cm}|M{2.4cm}|M{0.9cm}|M{0.9cm}|M{0.9cm}|}
\hline & & & & & & \\[-0.9em]
\multicolumn{1}{|c|}{\TE{no.}}&\multicolumn{1}{c|}{\TE{$(x_1 = m_1, a_1)$}}&\multicolumn{1}{c|}{\TE{$(x_2 = m_2, a_2)$}}&\multicolumn{1}{c|}{\TE{$(y = m', a')$}}&\multicolumn{1}{c|}{\TE{$\hjFVar{x}_1$}}&\multicolumn{1}{c|}{\TE{$\hjFVar{x}_2$}}&\multicolumn{1}{c|}{\TE{$\hjFVar{y}$}}\\ 
& & & & & & \\[-1em]
\specialrule{0.2em}{0.0em}{0.0em}  & & & & & & \\[-1em]
\TE{1} & \TE{$m_1<a_1$} & \TE{$m_2<a_2$} & \TE{$m'<a'$} & \TE{$\hjFVal{l}$} & \TE{$\hjFVal{l}$} & \TE{$\hjFVal{l}$} \\& & & & & & \\[-1em] \hline & & & & & & \\[-0.95em] 
\TE{2} & \TE{$m_1<a_1$} & \TE{$m_2=a_2$} & \TE{$m'<a'$} & \TE{$\hjFVal{l}$} & \TE{$\hjFVal{m}$} & \TE{$\hjFVal{l}$} \\& & & & & & \\[-1em] \hline & & & & & & \\[-0.95em] 
\TE{3} & \TE{$m_1<a_1$} & \TE{$m_2>a_2$} & \TE{$m'~?~a'$} & \TE{$\hjFVal{l}$} & \TE{$\hjFVal{h}$} & \TE{$\hjFVal{a}$} \\& & & & & & \\[-1em] \hline & & & & & & \\[-0.95em] 
\TE{4} & \TE{$m_1=a_1$} & \TE{$m_2<a_2$} & \TE{$m'<a'$} & \TE{$\hjFVal{m}$} & \TE{$\hjFVal{l}$} & \TE{$\hjFVal{l}$} \\& & & & & & \\[-1em] \hline & & & & & & \\[-0.95em] 
\TE{5} & \TE{$m_1=a_1$} & \TE{$m_2=a_2$} & \TE{$m'=a'$} & \TE{$\hjFVal{m}$} & \TE{$\hjFVal{m}$} & \TE{$\hjFVal{m}$} \\& & & & & & \\[-1em] \hline & & & & & & \\[-0.95em] 
\TE{6} & \TE{$m_1=a_1$} & \TE{$m_2>a_2$} & \TE{$m'>a'$} & \TE{$\hjFVal{m}$} & \TE{$\hjFVal{h}$} & \TE{$\hjFVal{h}$} \\& & & & & & \\[-1em] \hline & & & & & & \\[-0.95em] 
\TE{7} & \TE{$m_1>a_1$} & \TE{$m_2<a_2$} & \TE{$m'~?~a'$} & \TE{$\hjFVal{h}$} & \TE{$\hjFVal{l}$} & \TE{$\hjFVal{a}$} \\& & & & & & \\[-1em] \hline & & & & & & \\[-0.95em] 
\TE{8} & \TE{$m_1>a_1$} & \TE{$m_2=a_2$} & \TE{$m'>a'$} & \TE{$\hjFVal{h}$} & \TE{$\hjFVal{m}$} & \TE{$\hjFVal{h}$} \\& & & & & & \\[-1em] \hline & & & & & & \\[-0.95em] 
\TE{9} & \TE{$m_1>a_1$} & \TE{$m_2>a_2$} & \TE{$m'>a'$} & \TE{$\hjFVal{h}$} & \TE{$\hjFVal{h}$} & \TE{$\hjFVal{h}$} 
\\[0.15em]\hline 
\caption {Truth-table for $Add$}\label{STAdd}
\end{longtable}

\begin{definition}\label{Sub}
Let $Sub$ be the function defined by: $Sub(x_1, x_2) \Defs x_1{-}x_2$. Its associated FMB, FSub, is defined as follows:
\[
\begin{array}{ccc}
\{ \FVal{x}_1={\mathcal{h}}\lor \FVal{x}_2={l} \}&~~Sub~~& 
\{\FVal{y}= \mathcal{h}\}\\
\{ \FVal{x}_1=\mathcal{l}\lor \FVal{x}_2=\mathcal{h} \}&~~Sub~~& 
\{\FVal{y}= \mathcal{l}\}\\\\
\left\{\begin{array}{l}
 \FVal{x}_1=\mathcal{l} \land\FVal{x}_2=\mathcal{l} \lor\\
\FVal{x}_1=\mathcal{h} \land\FVal{x}_2=\mathcal{h} \lor \\
\FVal{x}_1=\mathcal{m} \land\FVal{x}_2=\mathcal{m} 
\end{array}\right\}
&~~ Sub~~& 
\{\FVal{y}= \mathcal{m}\}
\end{array}
\]
\end{definition}

\begin{longtable}[!h]{|M{0.7cm}|M{2.4cm}|M{2.4cm}|M{2.4cm}|M{0.9cm}|M{0.9cm}|M{0.9cm}|}
\hline & & & & & & \\[-0.9em]
\multicolumn{1}{|c|}{\TE{no.}}&\multicolumn{1}{c|}{\TE{$(x_1 = m_1, a_1)$}}&\multicolumn{1}{c|}{\TE{$(x_2 = m_2, a_2)$}}&\multicolumn{1}{c|}{\TE{$(y = m', a')$}}&\multicolumn{1}{c|}{\TE{$\hjFVar{x}_1$}}&\multicolumn{1}{c|}{\TE{$\hjFVar{x}_2$}}&\multicolumn{1}{c|}{\TE{$\hjFVar{y}$}}\\ 
& & & & & & \\[-1em]
\specialrule{0.2em}{0.0em}{0.0em}  & & & & & & \\[-1em]
\TE{1} & \TE{$m_1<a_1$} & \TE{$m_2<a_2$} & \TE{$m'~?~a'$} & \TE{$\hjFVal{l}$} & \TE{$\hjFVal{l}$} & \TE{$\hjFVal{a}$} \\& & & & & & \\[-1em] \hline & & & & & & \\[-0.95em] 
\TE{2} & \TE{$m_1<a_1$} & \TE{$m_2=a_2$} & \TE{$m'<a'$} & \TE{$\hjFVal{l}$} & \TE{$\hjFVal{m}$} & \TE{$\hjFVal{l}$} \\& & & & & & \\[-1em] \hline & & & & & & \\[-0.95em] 
\TE{3} & \TE{$m_1<a_1$} & \TE{$m_2>a_2$} & \TE{$m'<a'$} & \TE{$\hjFVal{l}$} & \TE{$\hjFVal{h}$} & \TE{$\hjFVal{l}$} \\& & & & & & \\[-1em] \hline & & & & & & \\[-0.95em] 
\TE{4} & \TE{$m_1=a_1$} & \TE{$m_2<a_2$} & \TE{$m'>a'$} & \TE{$\hjFVal{m}$} & \TE{$\hjFVal{l}$} & \TE{$\hjFVal{h}$} \\& & & & & & \\[-1em] \hline & & & & & & \\[-0.95em] 
\TE{5} & \TE{$m_1=a_1$} & \TE{$m_2=a_2$} & \TE{$m'=a'$} & \TE{$\hjFVal{m}$} & \TE{$\hjFVal{m}$} & \TE{$\hjFVal{m}$} \\& & & & & & \\[-1em] \hline & & & & & & \\[-0.95em] 
\TE{6} & \TE{$m_1=a_1$} & \TE{$m_2>a_2$} & \TE{$m'<a'$} & \TE{$\hjFVal{m}$} & \TE{$\hjFVal{h}$} & \TE{$\hjFVal{l}$} \\& & & & & & \\[-1em] \hline & & & & & & \\[-0.95em] 
\TE{7} & \TE{$m_1>a_1$} & \TE{$m_2<a_2$} & \TE{$m'>a'$} & \TE{$\hjFVal{h}$} & \TE{$\hjFVal{l}$} & \TE{$\hjFVal{h}$} \\& & & & & & \\[-1em] \hline & & & & & & \\[-0.95em] 
\TE{8} & \TE{$m_1>a_1$} & \TE{$m_2=a_2$} & \TE{$m'>a'$} & \TE{$\hjFVal{h}$} & \TE{$\hjFVal{m}$} & \TE{$\hjFVal{h}$} \\& & & & & & \\[-1em] \hline & & & & & & \\[-0.95em] 
\TE{9} & \TE{$m_1>a_1$} & \TE{$m_2>a_2$} & \TE{$m'~?~a'$} & \TE{$\hjFVal{h}$} & \TE{$\hjFVal{h}$} & \TE{$\hjFVal{a}$} 
\\[0.15em]\hline 
\caption {Truth-table for $Sub$}\label{STSub}
\end{longtable}

\begin{definition}\label{Abs}
Let $Abs$ be the function defined by: $Abs(x) \Defs |x|$. Its associated FMB, FAbs, is defined as follows:
\[
\begin{array}{ccc}
\{ \FVal{x}=\mathcal{l}\}&~~Abs~~& 
\{\FVal{y}= \mathcal{t}\}\\
\{ \FVal{x}=\mathcal{h} \}&~~Abs~~& 
\{\FVal{y}= \mathcal{f}\}\\\\
\{ \FVal{x}=\mathcal{l} \lor \FVal{x}=\mathcal{m} \lor \FVal{x}=\mathcal{h}\} 
&~~ Abs~~& 
\{\FVal{y}= \mathcal{m}\}
\end{array}
\]
\end{definition}

\begin{longtable}[!h]{|M{0.7cm}|M{2.4cm}|M{2.4cm}|M{0.9cm}|M{0.9cm}|}
\hline & & & & \\[-0.9em]
\multicolumn{1}{|c|}{\TE{no.}}&\multicolumn{1}{c|}{\TE{$(x = m, a)$}}&\multicolumn{1}{c|}{\TE{$(y = m', a')$}}&\multicolumn{1}{c|}{\TE{$\hjFVar{x}$}}&\multicolumn{1}{c|}{\TE{$\hjFVar{y}$}}\\ 
& & & & \\[-1em]
\specialrule{0.2em}{0.0em}{0.0em}  & & & & \\[-1em]
\TE{1} & \TE{$0 \leq m < a$} & \TE{$m'<a'$} & \TE{$\hjFVal{l}$} & \TE{$\hjFVal{l}$} \\& & & & \\[-1em] \hline & & & & \\[-0.95em] 
\TE{2} & \TE{$0 \leq m = a$} & \TE{$m'=a'$} & \TE{$\hjFVal{m}$} & \TE{$\hjFVal{m}$} \\& & & & \\[-1em] \hline & & & & \\[-0.95em] 
\TE{3} & \TE{$0 \leq a < m$} & \TE{$m'>a'$} & \TE{$\hjFVal{h}$} & \TE{$\hjFVal{h}$} \\& & & & \\[-1em] \hline & & & & \\[-0.95em] 
\TE{4} & \TE{$m < a < 0$} & \TE{$m'>a'$} & \TE{$\hjFVal{l}$} & \TE{$\hjFVal{h}$} \\& & & & \\[-1em] \hline & & & & \\[-0.95em] 
\TE{5} & \TE{$m = a < 0$} & \TE{$m'>a'$} & \TE{$\hjFVal{m}$} & \TE{$\hjFVal{h}$} \\& & & & \\[-1em] \hline & & & & \\[-0.95em] 
\TE{6} & \TE{$a < m < 0$} & \TE{$m'<a'$} & \TE{$\hjFVal{h}$} & \TE{$\hjFVal{l}$} \\& & & & \\[-1em] \hline & & & & \\[-0.95em] 
\TE{7} & \TE{$m < 0 \leq a$} & \TE{$m'~?~a'$} & \TE{$\hjFVal{l}$} & \TE{$\hjFVal{a}$} \\& & & & \\[-1em] \hline & & & & \\[-0.95em] 
\TE{8} & \TE{$a < 0 \leq m$} & \TE{$m'~?~a'$} & \TE{$\hjFVal{h}$} & \TE{$\hjFVal{a}$}  
\\[0.15em]\hline 
\caption {Truth-table for $Abs$}\label{STAbs}
\end{longtable}

\begin{definition}\label{Lcom}
Let $Lcom_K$ be the function defined by: $Lcom_K(x) \Defs (x < K)$. Its associated FMB, FLcom, is defined as follows:
\[
\begin{array}{ccc}
\{ \FVal{x}=\mathcal{l}\}&~~Lcom_K~~& 
\{\FVal{y}= \mathcal{t}\}\\
\{ \FVal{x}=\mathcal{h} \}&~~Lcom_K~~& 
\{\FVal{y}= \mathcal{f}\}\\\\
\{ \FVal{x}=\mathcal{l} \lor \FVal{x}=\mathcal{m} \lor \FVal{x}=\mathcal{h}\} 
&~~ Lcom_K~~& 
\{\FVal{y}= \mathcal{m}\}
\end{array}
\]
\end{definition}

\begin{longtable}[!h]{|M{0.7cm}|M{2.4cm}|M{0.9cm}|M{0.9cm}|M{0.9cm}|M{0.9cm}|}
\hline & & & & & \\[-0.9em]
\multicolumn{1}{|c|}{\TE{no.}}&\multicolumn{1}{c|}{\TE{$(x= m, a)$}}&\multicolumn{1}{c|}{\TE{$m'$}}&\multicolumn{1}{c|}{\TE{$a'$}}&\multicolumn{1}{c|}{\TE{$\hjFVar{x}$}}&\multicolumn{1}{c|}{\TE{$\hjFVar{y}$}}\\ 
& & & & & \\[-1em]
\specialrule{0.2em}{0.0em}{0.0em}  & & & & & \\[-1em]
\TE{1} & \TE{$K > m > a$} & \TE{$\True$} & \TE{$\True$} & \TE{$\hjFVal{h}$} & \TE{$\hjFVal{m}$}\\& & & & & \\[-1em] \hline & & & & & \\[-0.95em] 
\TE{2} & \TE{$K > m = a$} & \TE{$\True$} & \TE{$\True$} & \TE{$\hjFVal{m}$} & \TE{$\hjFVal{m}$}\\& & & & & \\[-1em] \hline & & & & & \\[-0.95em] 
\TE{3} & \TE{$K > a > m$} & \TE{$\True$} & \TE{$\True$} & \TE{$\hjFVal{l}$} & \TE{$\hjFVal{m}$}\\& & & & & \\[-1em] \hline & & & & & \\[-0.95em] 
\TE{4} & \TE{$m \geq K > a$} & \TE{$\False$} & \TE{$\True$} & \TE{$\hjFVal{h}$} & \TE{$\hjFVal{f}$}\\& & & & & \\[-1em] \hline & & & & & \\[-0.95em] 
\TE{5} & \TE{$a \geq K > m$} & \TE{$\True$} & \TE{$\False$} & \TE{$\hjFVal{l}$} & \TE{$\hjFVal{t}$}\\& & & & & \\[-1em] \hline & & & & & \\[-0.95em] 
\TE{6} & \TE{$m > a \geq K$} & \TE{$\False$} & \TE{$\False$} & \TE{$\hjFVal{h}$} & \TE{$\hjFVal{m}$}\\& & & & & \\[-1em] \hline & & & & & \\[-0.95em] 
\TE{7} & \TE{$m = a \geq K$} & \TE{$\False$} & \TE{$\False$} & \TE{$\hjFVal{m}$} & \TE{$\hjFVal{m}$}\\& & & & & \\[-1em] \hline & & & & & \\[-0.95em] 
\TE{8} & \TE{$a > m \geq K$} & \TE{$\False$} & \TE{$\False$} & \TE{$\hjFVal{l}$} & \TE{$\hjFVal{m}$} 
\\[0.15em]\hline 
\caption {Truth-table for $Lcom_K$}\label{STLcom}
\end{longtable}

\begin{definition}\label{Not}
Let $Not$ be the function defined by: $Not(x) \Defs \neg x$. Its associated FMB, FNot, is defined as follows:
\[
\begin{array}{ccc}
\{ \FVal{x}={\mathcal{t}}\}&~~Not~~& 
\{\FVal{y}= \mathcal{f}\}\\
\{ \FVal{x}={\mathcal{f}}\}&~~Not~~& 
\{\FVal{y}= \mathcal{t}\}\\\\
\left\{\begin{array}{l}
\FVal{x}=\mathcal{m} 
\end{array}\right\}
&~~ Not~~& 
\{\FVal{y}= \mathcal{m}\}
\end{array}
\]
\end{definition}

\begin{longtable}[!h]{|M{0.7cm}|M{0.9cm}|M{0.9cm}|M{0.9cm}|M{0.9cm}|M{0.9cm}|M{0.9cm}|}
\hline & & & & & & \\[-0.9em]
\multicolumn{1}{|c|}{\TE{no.}}&\multicolumn{1}{c|}{\TE{$m$}}&\multicolumn{1}{c|}{\TE{$a$}}&\multicolumn{1}{c|}{\TE{$m'$}}&\multicolumn{1}{c|}{\TE{$a'$}}&\multicolumn{1}{c|}{\TE{$\hjFVar{x}$}}&\multicolumn{1}{c|}{\TE{$\hjFVar{y}$}}\\ 
& & & & & & \\[-1em]
\specialrule{0.2em}{0.0em}{0.0em}  & & & & & & \\[-1em]
\TE{1} & \TE{\False} & \TE{\False} & \TE{\True} & \TE{\True} & \TE{$\hjFVal{m}$} & \TE{$\hjFVal{m}$} \\& & & & & & \\[-1em] \hline & & & & & & \\[-0.95em] 
\TE{2} & \TE{\False} & \TE{\True} & \TE{\True} & \TE{\False} & \TE{$\hjFVal{f}$} & \TE{$\hjFVal{t}$} \\& & & & & & \\[-1em] \hline & & & & & & \\[-0.95em] 
\TE{3} & \TE{\True} & \TE{\False} & \TE{\False} & \TE{\True} & \TE{$\hjFVal{t}$} & \TE{$\hjFVal{f}$} \\& & & & & & \\[-1em] \hline & & & & & & \\[-0.95em] 
\TE{4} & \TE{\True} & \TE{\True} & \TE{\False} & \TE{\False} & \TE{$\hjFVal{m}$} & \TE{$\hjFVal{m}$} 
\\[0.15em]\hline 
\caption {Truth-table for $Not$}\label{STNot}
\end{longtable}

\begin{definition}\label{And}
Let $And$ be the function defined by: $And(x_1, x_2) \Defs x_1 \wedge x_2$. Its associated FMB, FAnd, is defined as follows:
\[
\begin{array}{ccc}
\left\{\begin{array}{l} 
\FVal{x}_1={\mathcal{t}} \wedge \FVal{x}_2={t}\lor \\
\FVal{x}_1={\mathcal{t}} \wedge \FVal{x}_2={m} \lor \\
\FVal{x}_1={\mathcal{m}} \wedge \FVal{x}_2={t}
\end{array}\right\}
&~~And~~& 
\{\FVal{y}= \mathcal{t}\}\\\\
\left\{\begin{array}{l}
\FVal{x}_1=\mathcal{f}\lor \FVal{x}_2=\mathcal{f}
\end{array}\right\}
&~~And~~& 
\{\FVal{y}= \mathcal{f}\}\\\\
\left\{\begin{array}{l}
\FVal{x}_1=\mathcal{f} \land\FVal{x}_2=\mathcal{t} \lor\\
\FVal{x}_1=\mathcal{t} \land\FVal{x}_2=\mathcal{f} \lor\\
\FVal{x}_1=\mathcal{f} \land\FVal{x}_2=\mathcal{m} \lor\\
\FVal{x}_1=\mathcal{m} \land\FVal{x}_2=\mathcal{f} \lor\\
\FVal{x}_1=\mathcal{m} \land\FVal{x}_2=\mathcal{t} \lor\\
\FVal{x}_1=\mathcal{t} \land\FVal{x}_2=\mathcal{m} \lor\\
\FVal{x}_1=\mathcal{m} \land\FVal{x}_2=\mathcal{m} 
\end{array}\right\}
&~~ And~~& 
\{\FVal{y}= \mathcal{m}\}
\end{array}
\]
\end{definition}

\begin{longtable}[!h]{|M{0.7cm}|M{1.1cm}|M{1.1cm}|M{1.1cm}|M{1.1cm}|M{1.1cm}|M{1.1cm}|M{0.9cm}|M{0.9cm}|M{0.9cm}|}
\hline & & & & & & & & & \\[-0.9em]
\multicolumn{1}{|c|}{\TE{no.}}&\multicolumn{1}{c|}{\TE{$m_1$}}&\multicolumn{1}{c|}{\TE{$a_1$}}&\multicolumn{1}{c|}{\TE{$m_2$}}&\multicolumn{1}{c|}{\TE{$a_2$}}&\multicolumn{1}{c|}{\TE{$m'$}}&\multicolumn{1}{c|}{\TE{$a'$}}&\multicolumn{1}{c|}{\TE{$\hjFVar{x}_1$}}&\multicolumn{1}{c|}{\TE{$\hjFVar{x}_2$}}&\multicolumn{1}{c|}{\TE{$\hjFVar{y}$\footnote{Mode $\mathcal{u}$ in column $\hjFVar{y}$ is explained later.}}}\\ 
& & & & & & & & & \\[-1em]
\specialrule{0.2em}{0.0em}{0.0em}  & & & & & & &  & & \\[-1em]
\TE{1} & \TE{$\False$} & \TE{$\False$} & \TE{$\False$} & \TE{$\False$} & \TE{$\False$} & \TE{$\False$} & \TE{$\hjFVal{m}$} & \TE{$\hjFVal{m}$} & \TE{$\hjFVal{m}$} \\& & & & & & & & & \\[-1em] \hline & & & & & & & & & \\[-0.95em] 
\TE{2} & \TE{$\False$} & \TE{$\False$} & \TE{$\False$} & \TE{$\True$} & \TE{$\False$} & \TE{$\False$} & \TE{$\hjFVal{m}$} & \TE{$\hjFVal{f}$} & \TE{$\hjFVal{m}$} \\& & & & & & & & & \\[-1em] \hline & & & & & & & & & \\[-0.95em] 
\TE{3} & \TE{$\False$} & \TE{$\False$} & \TE{$\True$} & \TE{$\False$} & \TE{$\False$} & \TE{$\False$} & \TE{$\hjFVal{m}$} & \TE{$\hjFVal{t}$} & \TE{$\hjFVal{m}$} \\& & & & & & & & & \\[-1em] \hline & & & & & & & & & \\[-0.95em] 
\TE{4} & \TE{$\False$} & \TE{$\False$} & \TE{$\True$} & \TE{$\True$} & \TE{$\False$} & \TE{$\False$} & \TE{$\hjFVal{m}$} & \TE{$\hjFVal{m}$} & \TE{$\hjFVal{m}$} \\& & & & & & & & & \\[-1em] \hline & & & & & & & & & \\[-0.95em] 
\TE{5} & \TE{$\False$} & \TE{$\True$} & \TE{$\False$} & \TE{$\False$} & \TE{$\False$} & \TE{$\False$} & \TE{$\hjFVal{f}$} & \TE{$\hjFVal{m}$} & \TE{$\hjFVal{m}$} \\& & & & & & & & & \\[-1em] \hline & & & & & & & & & \\[-0.95em] 
\TE{6} & \TE{$\False$} & \TE{$\True$} & \TE{$\False$} & \TE{$\True$} & \TE{$\False$} & \TE{$\True$} & \TE{$\hjFVal{f}$} & \TE{$\hjFVal{f}$} & \TE{$\hjFVal{f}$} \\& & & & & & & & & \\[-1em] \hline & & & & & & & & & \\[-0.95em] 
\TE{7} & \TE{$\False$} & \TE{$\True$} & \TE{$\True$} & \TE{$\False$} & \TE{$\False$} & \TE{$\False$} & \TE{$\hjFVal{f}$} & \TE{$\hjFVal{t}$} & \TE{$\hjFVal{m}$} \\& & & & & & & & & \\[-1em] \hline & & & & & & & & & \\[-0.95em] 
\TE{8} & \TE{$\False$} & \TE{$\True$} & \TE{$\True$} & \TE{$\True$} & \TE{$\False$} & \TE{$\True$} & \TE{$\hjFVal{f}$} & \TE{$\hjFVal{m}$} & \TE{$\hjFVal{f}$} \\& & & & & & & & & \\[-1em] \hline & & & & & & & & & \\[-0.95em] 
\TE{9} & \TE{$\True$} & \TE{$\False$} & \TE{$\False$} & \TE{$\False$} & \TE{$\False$} & \TE{$\False$} & \TE{$\hjFVal{t}$} & \TE{$\hjFVal{m}$} & \TE{$\hjFVal{m}$} \\& & & & & & & & & \\[-1em] \hline & & & & & & & & & \\[-0.95em] 
\TE{10} & \TE{$\True$} & \TE{$\False$} & \TE{$\False$} & \TE{$\True$} & \TE{$\False$} & \TE{$\False$} & \TE{$\hjFVal{t}$} & \TE{$\hjFVal{f}$} & \TE{$\hjFVal{m}$} \\& & & & & & & & & \\[-1em] \hline & & & & & & & & & \\[-0.95em] 
\TE{11} & \TE{$\True$} & \TE{$\False$} & \TE{$\True$} & \TE{$\False$} & \TE{$\True$} & \TE{$\False$} & \TE{$\hjFVal{t}$} & \TE{$\hjFVal{t}$} & \TE{$\hjFVal{t}$} \\& & & & & & & & & \\[-1em] \hline & & & & & & & & & \\[-0.95em] 
\TE{12} & \TE{$\True$} & \TE{$\False$} & \TE{$\True$} & \TE{$\True$} & \TE{$\True$} & \TE{$\False$} & \TE{$\hjFVal{t}$} & \TE{$\hjFVal{m}$} & \TE{$\hjFVal{t}_u$} \\& & & & & & & & & \\[-1em] \hline & & & & & & & & & \\[-0.95em] 
\TE{13} & \TE{$\True$} & \TE{$\True$} & \TE{$\False$} & \TE{$\False$} & \TE{$\False$} & \TE{$\False$} & \TE{$\hjFVal{m}$} & \TE{$\hjFVal{m}$} & \TE{$\hjFVal{m}$} \\& & & & & & & & & \\[-1em] \hline & & & & & & & & & \\[-0.95em] 
\TE{14} & \TE{$\True$} & \TE{$\True$} & \TE{$\False$} & \TE{$\True$} & \TE{$\False$} & \TE{$\True$} & \TE{$\hjFVal{m}$} & \TE{$\hjFVal{f}$} & \TE{$\hjFVal{f}$} \\& & & & & & & & & \\[-1em] \hline & & & & & & & & & \\[-0.95em] 
\TE{15} & \TE{$\True$} & \TE{$\True$} & \TE{$\True$} & \TE{$\False$} & \TE{$\True$} & \TE{$\False$} & \TE{$\hjFVal{m}$} & \TE{$\hjFVal{t}$} & \TE{$\hjFVal{t}_u$} \\& & & & & & & & & \\[-1em] \hline & & & & & & & & & \\[-0.95em] 
\TE{16} & \TE{$\True$} & \TE{$\True$} & \TE{$\True$} & \TE{$\True$} & \TE{$\True$} & \TE{$\True$} & \TE{$\hjFVal{m}$} & \TE{$\hjFVal{m}$} & \TE{$\hjFVal{m}$} 
\\[0.15em]\hline 
\caption {Truth-table for $And$}\label{STAnd}
\end{longtable}

\begin{definition}\label{Or}
Let $Or$ be the function defined by: $Or(x_1, x_2) \Defs x_1 \lor x_2$. Its associated FMB, FOr, is defined as follows:
\[
\begin{array}{ccc}
\left\{\begin{array}{l}
\FVal{x}_1={\mathcal{t}} \lor \FVal{x}_2={t} 
\end{array}\right\}
&~~Or~~& 
\{\FVal{y}= \mathcal{t}\}\\\\
\left\{\begin{array}{l}
\FVal{x}_1=\mathcal{f}\wedge \FVal{x}_2=\mathcal{f} \lor \\
\FVal{x}_1={\mathcal{f}} \wedge \FVal{x}_2={m} \lor \\
\FVal{x}_1={\mathcal{m}} \wedge \FVal{x}_2={f}
\end{array}\right\}
&~~Or~~& 
\{\FVal{y}= \mathcal{f}\}\\\\
\left\{\begin{array}{l}
\FVal{x}_1=\mathcal{f} \land\FVal{x}_2=\mathcal{t} \lor\\
\FVal{x}_1=\mathcal{t} \land\FVal{x}_2=\mathcal{f} \lor\\
\FVal{x}_1=\mathcal{f} \land\FVal{x}_2=\mathcal{m} \lor\\
\FVal{x}_1=\mathcal{m} \land\FVal{x}_2=\mathcal{f} \lor\\
\FVal{x}_1=\mathcal{m} \land\FVal{x}_2=\mathcal{t} \lor\\
\FVal{x}_1=\mathcal{t} \land\FVal{x}_2=\mathcal{m} \lor\\
\FVal{x}_1=\mathcal{m} \land\FVal{x}_2=\mathcal{m} 
\end{array}\right\}
&~~ Or~~& 
\{\FVal{y}= \mathcal{m}\}
\end{array}
\]
\end{definition}

\begin{longtable}[!h]{|M{0.7cm}|M{1.1cm}|M{1.1cm}|M{1.1cm}|M{1.1cm}|M{1.1cm}|M{1.1cm}|M{0.9cm}|M{0.9cm}|M{0.9cm}|}
\hline & & & & & & & & & \\[-0.9em]
\multicolumn{1}{|c|}{\TE{no.}}&\multicolumn{1}{c|}{\TE{$m_1$}}&\multicolumn{1}{c|}{\TE{$a_1$}}&\multicolumn{1}{c|}{\TE{$m_2$}}&\multicolumn{1}{c|}{\TE{$a_2$}}&\multicolumn{1}{c|}{\TE{$m'$}}&\multicolumn{1}{c|}{\TE{$a'$}}&\multicolumn{1}{c|}{\TE{$\hjFVar{x}_1$}}&\multicolumn{1}{c|}{\TE{$\hjFVar{x}_2$}}&\multicolumn{1}{c|}{\TE{$\hjFVar{y}$}}\\ 
& & & & & & & & & \\[-1em]
\specialrule{0.2em}{0.0em}{0.0em}  & & & & & & &  & & \\[-1em]
\TE{1} & \TE{$\False$} & \TE{$\False$} & \TE{$\False$} & \TE{$\False$} & \TE{$\False$} & \TE{$\False$} & \TE{$\hjFVal{m}$} & \TE{$\hjFVal{m}$} & \TE{$\hjFVal{m}$} \\& & & & & & & & & \\[-1em] \hline & & & & & & & & & \\[-0.95em] 
\TE{2} & \TE{$\False$} & \TE{$\False$} & \TE{$\False$} & \TE{$\True$} & \TE{$\False$} & \TE{$\True$} & \TE{$\hjFVal{m}$} & \TE{$\hjFVal{f}$} & \TE{$\hjFVal{f}_u$} \\& & & & & & & & & \\[-1em] \hline & & & & & & & & & \\[-0.95em] 
\TE{3} & \TE{$\False$} & \TE{$\False$} & \TE{$\True$} & \TE{$\False$} & \TE{$\True$} & \TE{$\False$} & \TE{$\hjFVal{m}$} & \TE{$\hjFVal{t}$} & \TE{$\hjFVal{t}$} \\& & & & & & & & & \\[-1em] \hline & & & & & & & & & \\[-0.95em] 
\TE{4} & \TE{$\False$} & \TE{$\False$} & \TE{$\True$} & \TE{$\True$} & \TE{$\True$} & \TE{$\True$} & \TE{$\hjFVal{m}$} & \TE{$\hjFVal{m}$} & \TE{$\hjFVal{m}$} \\& & & & & & & & & \\[-1em] \hline & & & & & & & & & \\[-0.95em] 
\TE{5} & \TE{$\False$} & \TE{$\True$} & \TE{$\False$} & \TE{$\False$} & \TE{$\False$} & \TE{$\True$} & \TE{$\hjFVal{f}$} & \TE{$\hjFVal{m}$} & \TE{$\hjFVal{f}_u$} \\& & & & & & & & & \\[-1em] \hline & & & & & & & & & \\[-0.95em] 
\TE{6} & \TE{$\False$} & \TE{$\True$} & \TE{$\False$} & \TE{$\True$} & \TE{$\False$} & \TE{$\True$} & \TE{$\hjFVal{f}$} & \TE{$\hjFVal{f}$} & \TE{$\hjFVal{f}$} \\& & & & & & & & & \\[-1em] \hline & & & & & & & & & \\[-0.95em] 
\TE{7} & \TE{$\False$} & \TE{$\True$} & \TE{$\True$} & \TE{$\False$} & \TE{$\True$} & \TE{$\True$} & \TE{$\hjFVal{f}$} & \TE{$\hjFVal{t}$} & \TE{$\hjFVal{m}$} \\& & & & & & & & & \\[-1em] \hline & & & & & & & & & \\[-0.95em] 
\TE{8} & \TE{$\False$} & \TE{$\True$} & \TE{$\True$} & \TE{$\True$} & \TE{$\True$} & \TE{$\True$} & \TE{$\hjFVal{f}$} & \TE{$\hjFVal{m}$} & \TE{$\hjFVal{m}$} \\& & & & & & & & & \\[-1em] \hline & & & & & & & & & \\[-0.95em] 
\TE{9} & \TE{$\True$} & \TE{$\False$} & \TE{$\False$} & \TE{$\False$} & \TE{$\True$} & \TE{$\False$} & \TE{$\hjFVal{t}$} & \TE{$\hjFVal{m}$} & \TE{$\hjFVal{t}$} \\& & & & & & & & & \\[-1em] \hline & & & & & & & & & \\[-0.95em] 
\TE{10} & \TE{$\True$} & \TE{$\False$} & \TE{$\False$} & \TE{$\True$} & \TE{$\True$} & \TE{$\True$} & \TE{$\hjFVal{t}$} & \TE{$\hjFVal{f}$} & \TE{$\hjFVal{m}$} \\& & & & & & & & & \\[-1em] \hline & & & & & & & & & \\[-0.95em] 
\TE{11} & \TE{$\True$} & \TE{$\False$} & \TE{$\True$} & \TE{$\False$} & \TE{$\True$} & \TE{$\False$} & \TE{$\hjFVal{t}$} & \TE{$\hjFVal{t}$} & \TE{$\hjFVal{t}$} \\& & & & & & & & & \\[-1em] \hline & & & & & & & & & \\[-0.95em] 
\TE{12} & \TE{$\True$} & \TE{$\False$} & \TE{$\True$} & \TE{$\True$} & \TE{$\True$} & \TE{$\True$} & \TE{$\hjFVal{t}$} & \TE{$\hjFVal{m}$} & \TE{$\hjFVal{m}$} \\& & & & & & & & & \\[-1em] \hline & & & & & & & & & \\[-0.95em] 
\TE{13} & \TE{$\True$} & \TE{$\True$} & \TE{$\False$} & \TE{$\False$} & \TE{$\True$} & \TE{$\True$} & \TE{$\hjFVal{m}$} & \TE{$\hjFVal{m}$} & \TE{$\hjFVal{m}$} \\& & & & & & & & & \\[-1em] \hline & & & & & & & & & \\[-0.95em] 
\TE{14} & \TE{$\True$} & \TE{$\True$} & \TE{$\False$} & \TE{$\True$} & \TE{$\True$} & \TE{$\True$} & \TE{$\hjFVal{m}$} & \TE{$\hjFVal{f}$} & \TE{$\hjFVal{m}$} \\& & & & & & & & & \\[-1em] \hline & & & & & & & & & \\[-0.95em] 
\TE{15} & \TE{$\True$} & \TE{$\True$} & \TE{$\True$} & \TE{$\False$} & \TE{$\True$} & \TE{$\True$} & \TE{$\hjFVal{m}$} & \TE{$\hjFVal{t}$} & \TE{$\hjFVal{m}$} \\& & & & & & & & & \\[-1em] \hline & & & & & & & & & \\[-0.95em] 
\TE{16} & \TE{$\True$} & \TE{$\True$} & \TE{$\True$} & \TE{$\True$} & \TE{$\True$} & \TE{$\True$} & \TE{$\hjFVal{m}$} & \TE{$\hjFVal{m}$} & \TE{$\hjFVal{m}$} 
\\[0.15em]\hline 
\caption {Truth-table for $Or$}\label{STOr}
\end{longtable}

\paragraph{Practical implementation:}
For FMBs $FAnd$ and $FOr$ (Tables \ref{STAnd} and \ref{STOr}), we used subscript $\mathcal{u}$ in column $\hjFVar{y}$ to represent \emph{uncertain} faults \cite{Ref_200}. These are the scenarios for which a judgment on the propagation of fault cannot be made. Consider line 2 in Table $\ref{STOr}$ as an example. We know that the reported values at inputs $\Var{x}_1$ and  $\Var{x}_2$ and output $\Var{y}$ are all $\False$, and that matches up the function of $Or$. However, we cannot make a judgment as to which input dominates that state of the output; neither input $\False$ has priority over the other. This is different to the scenario in line 3, for instance, because the reported output is certainly dominated by the faulty input, which reports $\True$. For an FB output to be in a faulty state, not only the reported state of the output should be reachable by the given inputs, but its fault status should be a certain cause of a dominating input fault too. Otherwise, we cannot make a statement on whether and how the fault can propagate through the FB. 

In practice, uncertain failure scenarios such as rows 12 and 15 in Table \ref{STAnd} and rows 2 and 5 in Table \ref{STOr} will not be used for tracking failure modes in SIS programs. While this early filtering at the FB level may compromise the theoretical condition of completeness in some cases, it helps prevent producing combinations of non-faulty inputs as ``failure'' modes at the program level.  

Another filtering mechanism in FMR's backward tracking concerns the mode $\mathcal{m}$. During the tracking process, when we come across a variable with that particular mode, we do not continue tracking from that point any further. This is on the basis that, for whatever reasons, if a combination of input faults have resulted in no fault at the output, the combination is not important; i.e. the program output is immune against the failure combination in question.

In summary, the practical implementation of $FAnd$ and $FOr$ will include the following limited scenarios, which as can be seen are similar to the normal functions of AND and OR gates with two inputs:

$\\\begin{array}{ccc}
\left\{\begin{array}{l} 
\FVal{x}_1={\mathcal{t}} \wedge \FVal{x}_2={t}
\end{array}\right\}
&~~And~~& 
\{\FVal{y}= \mathcal{t}\}\\\\
\left\{\begin{array}{l}
\FVal{x}_1=\mathcal{f}\lor \FVal{x}_2=\mathcal{f}
\end{array}\right\}
&~~And~~& 
\{\FVal{y}= \mathcal{f}\}\\\\
\end{array}$

$\\\begin{array}{ccc}
\left\{\begin{array}{l}
\FVal{x}_1={\mathcal{t}} \lor \FVal{x}_2={t} 
\end{array}\right\}
&~~Or~~& 
\{\FVal{y}= \mathcal{t}\}\\\\
\left\{\begin{array}{l}
\FVal{x}_1=\mathcal{f}\wedge \FVal{x}_2=\mathcal{f}
\end{array}\right\}
&~~Or~~& 
\{\FVal{y}= \mathcal{f}\}\\\\
\end{array}$

\end{document}